\documentclass[10pt,onecolumn,oneside]{article}
\usepackage[english]{babel}
\usepackage{verbatim}
\usepackage{longtable}
\usepackage{graphicx}

\usepackage{url}
\PassOptionsToPackage{hyphens}{url}\usepackage{hyperref}

\makeatletter
\g@addto@macro{\UrlBreaks}{\UrlOrds}
\makeatother

\usepackage[font=small,skip=6pt]{caption}

\usepackage{amsmath,amssymb,amsfonts,amsthm,dsfont}

\theoremstyle{plain} \newtheorem{lemma}{\textbf{Lemma}}
\theoremstyle{plain} 
\theoremstyle{remark} \newtheorem{remark}{\textbf{Remark}}
\theoremstyle{plain} \newtheorem{theorem}{\textbf{Theorem}}
\theoremstyle{plain} 
\theoremstyle{plain} \newtheorem{corollary}{\textbf{Corollary}}
\theoremstyle{definition} \newtheorem{definition}{\textbf{Definition}}
\theoremstyle{plain} \newtheorem*{nonumbertheorem}{\textbf{Theorem}}
\theoremstyle{plain}

\usepackage{algorithm,algorithmicx,algpseudocode,float}

\makeatletter
\newenvironment{breakablealgorithm}
  {
   \begin{center}
     \refstepcounter{algorithm}
     \hrule height.8pt depth0pt \kern2pt
     \renewcommand{\caption}[2][\relax]{
       {\raggedright\textbf{\ALG@name~\thealgorithm} ##2\par}%
       \ifx\relax##1\relax 
         \addcontentsline{loa}{algorithm}{\protect\numberline{\thealgorithm}##2}%
       \else 
         \addcontentsline{loa}{algorithm}{\protect\numberline{\thealgorithm}##1}%
       \fi
       \kern2pt\hrule\kern2pt
     }
  }{
     \kern2pt\hrule\relax
   \end{center}
  }
\makeatother

\floatname{algorithm}{Procedure}

\renewcommand{\algorithmicreturn}[1]{\bgroup\\  ~#1\egroup}
\renewcommand{\algorithmiccomment}[1]{\bgroup\hfill//~#1\egroup}

\newcommand{\rmv}[1]{}

\DeclareMathOperator*{\esssup}{ess\,sup}

\newcommand*\rfrac[2]{{}^{#1}\!/_{#2}}
\newcommand{\squash}[1]{\raisebox{0.04ex}[0pt][0pt]{\small$\textstyle #1$}}

\makeatletter 
\let\@@pmod\pmod
\DeclareRobustCommand{\pmod}{\@ifstar\@pmods\@@pmod}
\def\@pmods#1{\mkern4mu({\operator@font mod}\mkern 6mu#1)}
\makeatother

\allowdisplaybreaks

\usepackage{authblk}
\usepackage{stackengine}
\title{Random variate generation using only finitely many unbiased, independently and identically distributed random bits}
\author{Luc Devroye}
\affil{\stackunder{\small{McGill University}}{\stackunder{\small{Canada}}{\mbox{\small{\texttt{lucdevroye@gmail.com}}}}}}
\author{Claude Gravel}
\affil{\stackunder{\small{EAGLYS Inc.}}{\stackunder{\small{Japan}}{\stackunder{\mbox{\small{\texttt{claudegravel1980@gmail.com}}}}{\mbox{\small{\texttt{c\_gravel@eaglys.co.jp}}}}}}}
\date{\today}

\begin{document}

\tableofcontents
\thispagestyle{empty}
\setcounter{page}{0}
\newpage

\maketitle

\begin{abstract}

For any discrete probability distributions with bounded entropy, we can generate exactly a random variate using only a finite expected number of perfect coin flips. A perfect coin flip is the outcome of an unbiased Bernoulli random variable. Coin flips are unbiased, independently and identically distributed in all our work. We survey well-known algorithms for the discrete case such as the one from Knuth and Yao as well as the one from Han and Hoshi. We also discuss briefly about a practical implementation for the algorithm proposed by Knuth and Yao. For the continuous case, only approximations can be hoped for. The freedom to choose the accuracy for the approximations matters, and, for that, we propose to measure accuracy in terms of the Wasserstein $L_\infty$-metric. We derive a universal lower bound for the expected number of perfect coin flips required to reach a desired accuracy. We also provide several algorithms for absolutely continuous distributions that come within our universal lower bound.



\end{abstract}

\section*{List of symbols}

Symbols and their short meanings used throughout this work are listed below. There might be some variants of the symbols mentioned below but we try to follow as much as possible the semantics given hereafter. Proper definitions are given in the following sections whenever required.

\begin{enumerate}
\item $\mathbf{p}$ is a probability vector.
\item $X$, $Y$, $Z$ or $U$ are random variables. Uppercase letters are generally used for random variables with the exceptions of $F$, $H$ and $W$.
\item $X_i$ is a sample or sequence of random variables for some $i\in I\subset \mathbb{Z}$.
\item $f$ is generally the probability density function of some random variable.
\item $F$ is generally the cumulative distributive function of some random variable.
\item $\stackrel{\mathrm{p}}{\to}$ denotes convergence in probability.
\item $\stackrel{\mathcal{D}}{=}$ denotes equality of distribution.
\item $\mathbf{P}$ is a generic symbol for a probability measure. Usage examples are: $\mathbf{P}\{X=x\}$, $\mathbf{P}\{X=x,Y=y\}$ or $\mathbf{P}\{X=x\mid Y=y\}$.
\item $\mathds{1}$ the indicator function. For instance $\mathds{1}\{X\in A\}$ is $0$ if $X\notin A$ and $1$ if $X\in A$.
\item $\mu_{i}$ or $\mu_{A}$ may be used as a generic symbol to render a more compact notation for $\mathbf{P}\{X=i\}$ (discrete) or $\mathbf{P}\{X\in A\}$ (continuous) respectively.
\item $H$ is generic symbol for the entropy of a distribution or the differential entropy. Usage examples are: $H(\mathbf{p})$, $H(X)$, $H(f)$ for the differential entropy, $H(X\mid Y=i)$ or $H(X\mid \mathcal{A})$ whenever $\mathcal{A}$ is a partition into disjoint sets of the support of $X$ which yields a discrete distribution. For the case of a continuous random variable $X$ and $\mathcal{A}$ a partition of its support, we sometime use the terminology of ``partition entropy'' to refer to $H(X\mid \mathcal{A})$.
\item $W_{p}$ is the Wasserstein $L_{p}$-metric defined over the product space of probability measures.
\item $\esssup$ is essential supremum which is the supremum excepted on sets of measure zero.
\item $\texttt{A}$ is a sampling algorithm.
\item $\texttt{RandomBit}$ is an instance of random bit generator.
\item $\texttt{FetchBit}$ is an almost-like instance of a random bit generator with the exception that it fetches random bits left by some processes and usually queued before calling $\texttt{RandomBit}$ as a subroutine. When the queue it is given to it is empty, it invokes $\texttt{RandomBit}$ automatically.
\item $T$ or $N$ are random variables for the complexities of interests (number of bits, stopping time, etc).
\item $f(x)^{+}=\max\{0,f(x)\}$ for a function $f:I\to\mathbb{R}$ and $I\subseteq\mathbb{R}^{d}$.
\end{enumerate}

\section{Introduction}\label{sect_intro}

Let $N$ be a discrete random variable with range $I\subseteq \mathbb{Z}$, and with distribution denoted by $p_i$ that is $p_i=\mathbf{P}\{N=i\}$ for $i\in I$. The binary entropy of $N$ is the quantity $-\sum_{i\in I}{p_{i}\log_{2}(p_i)}=H(N)$. To generate a random variate, we assume the existence of a source of random unbiased bits. More specifically, we denote by \texttt{RandomBit} a device, a method, or an oracle that is assumed to return an unbiased bit independently of any previous calls when invoked. We do not discuss how to create software or hardware instances of \texttt{RandomBit} here.  A random unbiased bit is a Bernoulli random variable with equal probability for either of its two outcomes. Knuth and Yao \cite{KnuYao_1976} showed that the expected number of independent unbiased random bits needed to generate an instance of $N$ is at least equal to the binary entropy of $N$. They also exhibited an algorithm called the Discrete Data Generator tree algorithm, abbreviated DDG tree hereafter, for which the expected number of random unbiased bits is not more than $H(N)+2$. Another famous DDG based tree algorithm appeared later from Han and Hoshi \cite{HanHos_1997}. DDG based tree algorithms rely on the perfect knowledge of the probability vector $(p_i)_{i\in I}$ and therefore assumes a computational capability with arbitrary precision over real numbers. In some cases such as the discrete uniform distribution or some other distributions with particular structures, there is no need to perform computations with arbitrary finite precision for the probabilities. For instance, Lumbroso \cite{Lum_PHD_2012} created an algorithm to sample discrete uniform distributions with an expected complexity that fits within the information theoretical interval provided by Knuth and Yao, and for which only \emph{integer} arithmetic is required. We mention briefly later a simple and practical implementation of the algorithm from Knuth and Yao. From a theoretical point of view, we also exhibit an interesting batch-type algorithm which has asymptotically the binary entropy as expected complexity; the key idea of our algorithm is to extract random bits left in the generation process based on DDG trees. We may use both the terms variate and variable interchangeably.

While the aforementioned results settle the discrete random variate case quite satisfactorily, the generation of continuous or mixed random variables has not been treated satisfactorily in the literature. One of our goal is to study the expected number of random unbiased bits to generate a continuous variate $X\in\mathbb{R}^{d}$ from a continuous distribution with a given precision or accuracy $\epsilon>0$. A few important concepts to recall, and upon which rely the definition of a sampling algorithm, are the $\ell_{p}$-norm and the differential entropy of an absolutely continuous probability distribution. For a vector $v\in\mathbb{R}^{d}$, let $\|v\|_{p}$ denote the $\ell_{p}$-norm of $v$ for $p\geq 1$:
$\|v\|_{p}=\big(\sum_{i=1}^{d}{|v_{i}|^{p}}\big)^{1/p}$. For $p=\infty$, the $\infty$-norm is $\|v\|_{\infty}=\sup_{1\leq i\leq d}{|v_{i}|}$. With $d=1$, all $p$-norms are the same for $p\in[1,\infty]$. If the distribution of $X$ is absolutely continuous with density $f$ on a support $S\subseteq \mathbb{R}^{d}$, then we denote by $H(f)$ the differential entropy of $f$ (or $X$) which is given by
\begin{displaymath}
H(f) = \int_{S}{f(x)\log_{2}\bigg(\frac{1}{f(x)}\bigg)dx}.
\end{displaymath}
The differential entropy can be ill-defined, $-\infty$, finite or $+\infty$. We refer to Cover and Thomas \cite{CovTho_book1991} for more information on differential entropy and entropy in general. When $f$ has a compact support, then the case $+\infty$ cannot occur. When $f$ is bounded, then the case $-\infty$ is excluded. When $H(\lfloor X_{1}\rfloor, \ldots, \lfloor X_{d}\rfloor) < \infty$, it can be shown that $H(f)$ is well-defined and is either finite or $-\infty$; see R\'{e}nyi \cite{Ren_1959}, Csisz\`{a}r \cite{Csi_1961} for a proof.

A satisfactory choice of metric to measure the accuracy is the Wasserstein $L_\infty$-metric between two probability measures. The Wasserstein $L_{p}$-metrics are explained in details in Rachev and R{\"u}schendorf \cite{RacRus_volI} and \cite{RacRus_volII}. Let $\mathcal{M}$ denote the product space of all distributions of pairs \mbox{$(X,Y) \in \mathbb{R}^{d}\times\mathbb{R}^{d}$} with fixed marginal distributions $F$ and $G$ for $X$ and $Y$, respectively. Then the Wasserstein $L_\infty$-distance between $X$ and $Y$, or between $F$ and $G$, is
\begin{displaymath}
\begin{split}
W_{p}(F,G)=\inf\big\{\esssup\|X-Y\|_{p}\phantom{1}:\phantom{1}(F,G)\in \mathcal{M}\},
\end{split}
\end{displaymath}
where $\esssup$ denotes the essential supremum. The Wasserstein $L_\infty$-metric defines a distance between $X$ and $Y$ that is $\mathrm{dist}_{p}(X,Y)= W_{p}(F,G)$. If $\mathrm{dist}_{p}(X,Y) < \epsilon$, then there exists a random variable $Y$ (output) coupled with $X$ (target) such that \mbox{$\esssup\|X-Y\|_{p} < \epsilon$} that is, with probability one, $\|X-Y\|_{p}<\epsilon$. We go beyond the existence of $Y$ and show in later sections how to generate such instances of $Y$. This definition of distance satisfies almost all simulation scenarios that require the evaluation of a \emph{continuous} real-valued function $\Psi(X_{1},\ldots,X_{d})$ where the $X_i$'s are independent random variables. For $X,Y,y\in\mathbb{R}^{d}$, then, almost surely, we have $|\Psi(Y)-\Psi(X)|\leq\sup\{|\Psi(y)-\Psi(X)|\colon \|y-X\|_{p}<\epsilon\}$ which can be controlled by the user.

In the following definition, let $\epsilon$ be the accuracy between a desired target random variable $X$ and the output $Y$ from a generation algorithm. Also $T$ is a random variable that denotes the number of times \texttt{RandomBit} is invoked by a algorithm that generates a random instance of $Y$.
\begin{definition}[$(\epsilon,p)$-sampling algorithm]\label{def_arbsamalg}
On inputs $\epsilon$ and $p$, an $(\epsilon,p)$-sampling algorithm $\mathtt{A}$ for $X$ is a probabilistic algorithm that returns $Y$ such that $\|X-Y\|_{p}<\epsilon$ with probability one and halts when \texttt{RandomBit} is invoked $T$ times.
\end{definition}
We are interested in sampling algorithms for which the expectation of the stopping time $T$ from definition \ref{def_arbsamalg} is finite. One of our main result is the following:
\begin{nonumbertheorem}
Let $X \in \mathbb{R}^{d}$ be a random vector with density $f$, and assume that the entropy of the integer parts of the components of $X$ is finite, that is, \mbox{$H(\lfloor X\rfloor)<\infty$} where $X=(\lfloor X_{1}\rfloor, \ldots, \lfloor X_{d}\rfloor)$. The expected number of random i.i.d.~unbiased bits, $\mathbf{E}(T)$, used by any sampling algorithm for $X$ and output accuracy $\epsilon$ is bounded below by
\begin{displaymath}
H(f)+d\log_{2}\bigg(\frac{1}{\epsilon}\bigg)-\log_{2}V_{d,p}\quad\text{with}\quad V_{d,p}=\frac{2^{d}\Gamma\big(\frac{1}{p}+1\big)}{\Gamma\big(\frac{d}{p}+1\big)},
\end{displaymath}
and the latter quantity is the volume of the unit ball in $\mathbb{R}^{d}$. For $p=1$ and $p=\infty$, the third term in the lower bound is $\log_{2}\big(2^{d}\slash d!\big)$ and $d$, respectively. For $d=1$, it is $1$.
\end{nonumbertheorem}

We provide most importantly the foundational background to research universal lower bounds for the generation of continuous random variate with arbitrary finite precision and finite expected complexity. We also provide a methodology for various useful upper bounds for practical algorithms. Among those practical algorithms is the one by Devroye and Gravel \cite{DevGra_2017} which is an extension of the Von Neumann's rejection method to our realistic practical framework. For the sampling of absolutely continuous distributions, we observe that $\mathbf{E}(T)$ relates to the binary entropy almost in the way done in Knuth and Yao \cite{KnuYao_1976} for the discrete case. Some authors have addressed the problem of arbitrary finite precision for sampling algorithms for continuous distributions. Among them, Flajolet and Saheb \cite{FlaSah_1986} explain how to generate the first $k$ bits of an exponential random variable for an integer $k\geq 1$, and Karney \cite{Kar_2016} describes an algorithm for the standard normal distribution.

This article is divided into two major sections: section \ref{sect_disc_var_gen} treats the discrete case and section \ref{sect_main_sect_cts} treats the continuous distributions. These two main sections start with a brief summary of their content.

\section{Discrete variate generation}\label{sect_disc_var_gen}

We discuss here the two main approaches to generate discrete random variables: the Knuth and Yao \cite{KnuYao_1976} algorithm and the Han and Hoshi \cite{HanHos_1997} algorithm. Both former algorithms make extensive use of trees as data structures, and particularly, a type of tree called Discrete Data Generator tree or DDG-tree for short. In section \ref{sect_KY_algo}, we recall the Knuth and Yao's algorithm which encode a probability mass function into a DDG tree. In section \ref{sect_HH}, we recall Han and Hoshi's algorithm which encodes a cumulative distribution function into a DDG tree. If $\mathbf{p}$ denotes the target probability vector to be sampled, then both of the previous algorithms needs an expected number of random bits of about $H(\mathbf{p})+O(1)$ given that $H(\mathbf{p})$ is bounded.

In section \ref{sect_batch_gen}, we develop further and generalize the concept of a DDG-tree based algorithm, and, from there, it naturally follows our main contribution to the discrete case found in section \ref{sect_assymptotic_batch}: an algorithm that reaches the Shannon's lower bound \cite{Sha_1948}, that is $H(\mathbf{p})$. To reach asymptotically and in probability Shannon's bound, we develop a method in section \ref{sect_rnd_ext} to extract randomness from i.i.d.\ random variables distributed according to some distributions. A batch is a sequential generation of i.i.d.\ random variables. The randomness extraction procedure is then used within our asymptotic batch generation method from \ref{sect_assymptotic_batch}.

In section \ref{sect_lumbroso_fast_dice}, we recall Lumbroso's algorithm \cite{Lum_PHD_2012} to generate a discrete uniform distribution. In section \ref{sect_concrete_imple}, we explain a C++ implementation for the Knuth and Yao's algorithms for general mass functions, other than just the uniform distribution, that uses lists as data structures, and which shares lots of similarities with Lumbroso's algorithm.

\subsection{DDG tree algorithm and probability mass function}\label{sect_KY_algo}

In this section, we detail principles and facts behind the Knuth and Yao \cite{KnuYao_1976} sampling algorithm. The concept of a DDG-tree as a data structure is central in order to encode a probability vector and obtain an almost optimal expected complexity. The interval for the expected complexity of the number of random bits is $[H,H+2]$ where $H$ is the binary entropy of the discrete distribution to be sampled whenever $H$ is finite. By the result of Shannon \cite{Sha_1948}, the expected complexity must be at least $H$ bits.

For $I\subseteq \mathbb{Z}$, let $\mathbf{p}=(p_{i})_{i\in I}$ be a probability vector, that is, $p_i>0$ for all $i\in I$ and $\sum_{i\in I}{p_{i}}=1$. For $i\in I$, we write the binary expansion of $p_i$ as
\begin{displaymath}
p_{i}=\sum_{j=1}^{\infty}{p_{ij}2^{-j}}\quad\text{for $p_{ij}\in\{0,1\}$}.
\end{displaymath}
For a while, suppose we have the ability to compute $p_{ij}$ on the fly or the ability of infinite storage whenever $p_i$ are irrational numbers. For $j\geq 1$, consider the family of sets (lists) $L_{j}$ defined by $L_{j}=\{i\in I\text{ and }p_{ij}=1\}$. In other words, $L_j$ is the set of outcomes which have non-zero coefficient for the term $2^{-j}$ in their probability of occurrence. We have that
\begin{displaymath}
\sum_{i\in I}{p_{i}}=\sum_{i\in I}\sum_{j=1}^{\infty}{p_{ij}2^{-j}}=\sum_{j=1}^{\infty}\sum_{i\in I}{p_{ij}2^{-j}}=\sum_{j=1}^{\infty}{\frac{|L_j|}{2^{j}}}=1.
\end{displaymath}
Clearly $0\leq |L_{j}|\leq 2^{j}$ for all $j\geq 1$. More importantly, $L_j$ is \emph{uniformly} distributed that is $\mathbf{P}\{i\in L_j\}=1/|L_{j}|$. We observe that the only case for which $|L_{j'}|=0$ for some $j>j'$ corresponds to the uniform distribution with $2^{j'}$ atoms. An atom is an element from the support of a discrete distribution.

The family of lists $L=\{L_{j}\}_{j=1}^{\infty}$ defines \emph{uniquely} a tree that Knuth and Yao termed the Discrete Data Generator tree, abbreviated DDG tree. We can add a member to $L$, namely $L_0=\emptyset$, to represent the root of the tree and $|L_0|=0$.
A probability vector $\mathbf{p}=(p_i)_{i\in I}$ has a unique (often of infinite size) DDG binary tree representation for which
\begin{enumerate}
\item[(1)] leaves with depth $j$ are the elements of $L_j$,
\item[(2)] the number of nodes with depth $j$ that are not leaves is denoted by $s_j$ and equals $t_j-|L_j|$ where $t_j$ is the total number of nodes with depth $j$.
\end{enumerate}
Without loss of generality, $L_0=\emptyset$ and we have for all $j\geq 1$ that
\begin{displaymath}
t_{j}=2^{j}-\sum_{k=0}^{j-1}{2^{j-k}|L_{k}|}\quad\text{and}\quad t_j=s_{j}+|L_{j}|.
\end{displaymath}
A visual example may help and let us consider for instance $\mathbf{p}=(p_0,p_1,p_2)$ where
\begin{align}
p_{0} & = \frac{1}{\pi} = (0.010100010111110\ldots)_{2},\label{onara_irrationnel1}\\
p_{1} & = \frac{1}{e} = (0.010111100010110\ldots)_{2},\label{onara_irrationnel2}\\
p_{2} & = 1-p_{1}-p_{2} = (0.010100000101010\ldots)_{2}\label{onara_irrationnel3}.
\end{align}
The DDG tree of $\mathbf{p}$ has infinite size and is represented on figure \ref{figDDG}. Elements in $L_{j}$ are all equally likely, and their indexing as leaves on the corresponding level does not matter as well. However it is custom to take same canonical order.

\begin{figure}
\begin{centering}
\includegraphics{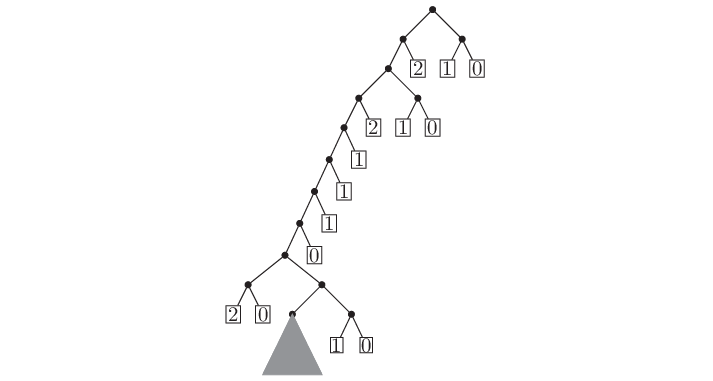}
\caption{\label{figDDG}DDG tree for a probability vector of length $3$ with irrational probabilities given by Equations (\ref{onara_irrationnel1}), (\ref{onara_irrationnel2}) and (\ref{onara_irrationnel3}). The gray triangle means the tree continues ad aeternam.}
\end{centering}
\end{figure}

Given the ability to generate uniform i.i.d.~bits and given a (non)-uniform discrete distributions $\mathbf{p}=(p_1,p_2\ldots)$, a random walk beginning from the root of the tree to a leave with depth $j$ generates an instance of the latter distribution. We adopt the convention to walk to the left when a random bit $0$ is returned by \texttt{RondomBit} and to the right when $1$ is returned. The algorithm halts with probability one. We can give other meanings to the quantities $t_j$, $|L_j|$ and $s_j$ as follow:
\begin{enumerate}
\item[(1)] $t_j$ is the number of decisions at depth $j$,
\item[(2)] $|L_j|$ is the number of ways the random walk stops at depth $j$. Given the walk has reached depth $j$, then it stops with probability $|L_j|/t_j$,
\item[(3)] $s_j$ is the number of ways the random walk continues to depth $j+1$. Given the walk has reached level $j$, then it continues with probability $s_j/t_j$.
\end{enumerate}
As shown in \cite{KnuYao_1976}, the former quantities entirely characterized the expected number of bits and hence the running time or stopping time of the random walk over the DDG tree.

To generate a random outcome given the knowledge of $L=\{L_j\}_{j\geq 1}$ and given an access to an instance of \texttt{RandomBit}, we use the latter to obtain \emph{uniform} random values in the intervals $[0,t_1)$, $[0,t_2)$, $\ldots$, $[0,t_j)$, $\ldots$, and stops as soon as the latter uniform value is in $[0,|L_{j}|)$. We have therefore the following algorithm in which $L_{j}[x]$ means the $x$-th member of $L_{j}$ and $|L_{j}|$ stands for the size of $L_{j}$.

\begin{breakablealgorithm}\label{alg_ky_by_lists}
\caption{The Knuth and Yao's sampling algorithm using lists}
\begin{algorithmic}[1]
\raggedright
\Require Lists $L_{j}$ for $j\geq 1$
\Ensure $X$ distributed according to $\mathbf{p}$
\State{$X\leftarrow 0$}
\State{$Y\leftarrow 1$}
\State{$j\leftarrow 1$}
\Loop
\State{$Y\leftarrow 2Y$}
\State{$X\leftarrow 2X+\texttt{RandomBit}$}
\If{$Y\geq t_{j}$}
\If{$X<|L_{j}|$}
\State{\textbf{Return} $L_{j}[X]$}
\Else
\State{$Y\leftarrow Y-|L_{j}|$}
\State{$X\leftarrow X-|L_{j}|$}
\EndIf
\EndIf
\State{$j\leftarrow j+1$}
\EndLoop
\end{algorithmic}
\end{breakablealgorithm}

Algorithm \ref{alg_ky_by_lists} is correct from the discussions preceding its elaboration. We give now a proof of its expected complexity.
\begin{theorem}[Knuth and Yao \cite{KnuYao_1976}]\label{KYthm}
The expected number of calls to \texttt{RandomBit} in algorithm \ref{alg_ky_by_lists} is bounded from below by $H(\mathbf{p})$ and from above by $H(\mathbf{p})+2$.
\end{theorem}

\begin{proof}[Proof of theorem \ref{KYthm}]
Given a probability vector $\mathbf{p}=(p_1,p_2,\ldots,p_n)$ with $n$ possibly infinite, recall the binary expansion of $p_i$ which is
\begin{align*}
p_i&=\sum_{j=1}^{\infty}{\frac{p_{ij}}{2^{j}}}\quad\text{for $p_{ij}\in\{0,1\}$}.
\end{align*}
If $T$ denotes the number of bits required by the random walk to sample $\mathbf{p}$, then for $j\geq 1$
\begin{align}
\mathbf{P}\{T=j\}&=\frac{\textrm{number of leaves at level $j$}}{2^{j}}=\sum_{i=1}^{n}{\frac{p_{ij}}{2^{j}}}=\sum_{i=1}^{n}{\frac{|L_{j}|}{2^{j}}},\nonumber\\
\mathbf{E}(T)&=\sum_{j=0}^{\infty}{t\mathbf{P}\{T=j\}}=\sum_{j=1}^{\infty}{t\sum_{i=1}^{n}{\frac{p_{ij}}{2^{j}}}}=\sum_{i=1}^{n}{\bigg(\sum_{j=1}^{\infty}{\frac{jp_{ij}}{2^{{j}}}}\bigg)}\label{eqn:expcomp_foo}.
\end{align}
We focus now on the quantity between parentheses from (\ref{eqn:expcomp_foo}). For that, let $m_i$ be the smallest integer such that $p_{i{}m_i}=1$, and $p_{ij}=0$ for $1\leq j\leq m_{i}-1$, then
\begin{displaymath}
m_i=\bigg\lceil\log_{2}\bigg(\frac{1}{p_i}\bigg)\bigg\rceil=\log_{2}\bigg(\frac{1}{p_i}\bigg)+\delta_i\quad\text{for some $0\leq\delta_i<1$.}
\end{displaymath}
The quantity within parentheses of (\ref{eqn:expcomp_foo}) is bounded above by
\begin{displaymath}
\sum_{j=m_i}^{\infty}{\frac{j}{2^{j}}}=\frac{m_i+1}{2^{m_i-1}}=\bigg(\log_{2}\bigg(\frac{1}{p_i}\bigg)+\delta_i+1\bigg)p_{i}2^{1-\delta_i}= \vartheta_{i}
\end{displaymath}
The first derivative of $\vartheta_{i}$ with respect to $\delta_i$ vanishes only when $\delta_i=\log_{2}\big(\rfrac{ep_i}{2}\big)$, and its second derivative is positive when $\delta_i=\log_{2}\big(\rfrac{ep_i}{2}\big)$. Therefore we analyze the quantity $\vartheta_{i}$ with respect to the three possible cases: (1) when $p_i>\rfrac{4}{e}$ (impossible since $p_i<1$ by definition), (2) when $\rfrac{2}{e}<p_i\leq 1$ (more generally when $\rfrac{2}{e}<p_i\leq\rfrac{4}{e}$, but again $p_i<1$ by definition), and (3) when $p_i\leq \rfrac{2}{e}$. If $\rfrac{2}{e}<p_i\leq 1$, then the minimum of $\vartheta_{i}$ occurs within the interval $(0,1]$, and $\vartheta_{i}<2$ for both the boundaries. If $p_i\leq\rfrac{2}{e}$, then $\vartheta_{i}\leq p_i\log_{2}\big(\rfrac{1}{p_i}\big)+2p_i$ since the maximum occurs at the right boundary $\delta_i=1$ since, in this case, $\vartheta_{i}$ is increasing on $(0,1]$. By summing over the $i$'s, one obtains the upper bound.

Hence given $i$, the quantity within parentheses of (\ref{eqn:expcomp_foo}) is bounded below from
\begin{displaymath}
\sum_{j=m_i}^{\infty}{\frac{j{p_{ij}}}{2^{j}}}\geq m_i\sum_{j=m_i}^{\infty}{\frac{p_{ij}}{2^j}}=m_{i}p_{i}>p_{i}\log_{2}\bigg(\frac{1}{p_{i}}\bigg).
\end{displaymath}
By summing over the $i$'s the quantity , we obtain the lower bound.
\end{proof}

\subsubsection{A brief discussion about a concrete implementation}\label{sect_concrete_imple}

Given a mathematical representation for $\mathbf{p}=(p_i)_{i\in I}$, we can wonder how to obtain the corresponding lists $L_{j}$'s for sufficiently large $j$'s such that, in practice, an implementation of algorithm \ref{alg_ky_by_lists} halts with probability as close to one as the storage allows. When implementing algorithm \ref{alg_ky_by_lists} in practice, we need to take into account that storage is finite, and, more importantly, how storage and accuracy (entropy) relates to each other. The exactness of the lists relies on the ability of libraries to perform exact arithmetic, and this is itself a topic in computational arithmetic number theory. What matters of a library is the guarantee to obtain truncations of $p_i$ with sufficiently enough bits so that an implementation of \ref{alg_ky_by_lists} halts with probability as close to one as the storage allows. We observe that generally $p_i$ is an irrational aperiodic number and therefore truncations with sufficiently many bits from its binary expansion are needed to halt.

Computing libraries such as \cite{GMP}, \cite{MPFR} or \cite{NTL} can compute functions and perform arithmetic operations with guaranteed accuracy. For instance, outputs from GMP are always truncated to the destination variable’s precision. MPFR is a GMP extension for multiple-precision floating-point computations with correct rounding. MPFR provides well-defined precision and accurate rounding, and thereby naturally extends IEEE P754. NTL provides a module to represent arbitrary-precision floating point
numbers. The functions from that NTL module guarantee very strong accuracy conditions which make it easy to reason about the behavior of programs using these functions. The arithmetic operations always round their results to the current precision.

The storage required for the lists $L_j$ depends on a few factors that we briefly mention. First we need to compute and store at least $\lceil H(\mathbf{p})\rceil=j_1$ lists $L_{j}$'s for $1\leq j\leq j_1$. The quantity $j_1$ is an average case complexity. For worst-case scenarios, another choice is $j_1=\lceil-\log_{2}(\min\{p_i\colon i\in I\})\rceil$ given $I$ is finite. Note that in all of our work, we always have assumed $p_i>0$ so that the former and latter quantities are well-defined. Second when $\mathbf{p}$ has infinite support which is truncated to a finite support which must yield to a properly normalized probability vector $\mathbf{q}$; let us write $\delta$ for the leftover probability due to the truncation. Then in addition to the first $j_1$ lists, we need $\lceil -\log_{2}(1-\delta)\rceil=j_2$ lists. If $\mathbf{p}$ has a finite support, then usually no truncation of the support is required so that $\delta=0$ and no additional list are required. Third, it is desirable often to make sure that an implementation is indistinguishable from the ideal target. For that matter, an additional number of lists, say $j_3$, should be computed. This $j_3$ additional lists mostly depends on the computational power to perform statistical goodness-of-fit tests. Stein's theorem for instance can be helpful in bounding the quantity $j_3$ here.

For instance, the second author’s GitHub \url{https://github.com/63EA13D5/} contains a C++ implementation of algorithm \ref{alg_ky_by_lists}. The implementation is rather straightforward and uses only classes from the C++ standard library excepted for the exact computations of the probabilities upon which the implementation depends on NTL \cite{NTL} which itself depends on \cite{MPFR} for its class on exact arithmetic over real numbers. We point out, that for a good implementation, the wall time to execute algorithm \ref{alg_ky_by_lists} is linearly proportional to the expected complexity. The ratio of the wall time by the number of random coins needed solely depends on the machine architecture. Also, for many libraries such as NTL, the accuracy required for exact arithmetic must be determined at the time of compilation. Different libraries may use different arithmetic methods and more accuracy may be required through intermediate or auxiliary computations.

In order to give two simple examples, let us take the case of the binomial and the Zeta-Dirichlet distributions. We recall that if a random variable $X$ has a binomial distribution with parameters $N$ and $p$ as the number of trials and occurrence probability, respectively, then $\mathbf{P}\{X=i\}=\binom{N}{i}p^{i}(1-p)^{N-i}$ for $0\leq x\leq N$. If $X$ is distributed has a Zeta-Dirichlet with concentration parameter $u>0$, then $\mathbf{P}\{X=i\}=\big(C_u{}i(\log(i))^{1+u}\big)^{-1}$ for $i\geq 3$ and $C_u$ is the normalization constant. In the case of the Zeta-Dirichlet, we truncate its support at $i=10000$ for our needs and re-normalize using sufficiently enough accuracy. In table \ref{bino_ky_implementaton_prog_experiment_data} and \ref{trunc_zeta_dir_ky_implementaton_prog_experiment_data}, column titles are abbreviated by ``Par.'', ``Ent.'', ``Emp. est. $\mathbf{E}(T)$'', and ``Ave. time gen.'' and stands for parameters, entropy, empirical estimation of $\mathbf{E}(T)$, and average time generation, respectively. The empirical estimation of the expected number of coin flips, denoted here $\mathbf{E}(T)$, is the average over the sample of the number of calls to our instance of \texttt{RandomBit}. The empirical estimation of $\mathbf{E}(T)$ has to be compared with the theoretical entropy, and must not differs by $2$ bits above the entropy as a bare criteria to check the correctness as implied by theorem \ref{KYthm}. The average time generation is the average over the sample of the wall time to generate random outcomes in the sample; it is given in milliseconds for the sake of completeness and many architectural factors of the hardware for instance influence it.

\begin{centering}
\begin{longtable}{|c|c|c|c|}\caption{Empirical complexities for a sample of size $100000$ of binomial random variables}\label{bino_ky_implementaton_prog_experiment_data}\\ \hline
Par. $(N,p)$ & Ent. & Emp. est. $\mathbf{E}(T)$ & Ave. time gen. (ms)\\ \hline
\endfirsthead
Par. $(N,p)$ & Ent. & Emp. est. $\mathbf{E}(T)$ & Ave. time gen. (ms)\\ \hline
\endhead
\multicolumn{4}{c}{Continued on next page}
\endfoot
\endlastfoot
\hline
$(100,0.005)$ & $1.337262$ & $2.278150$ & $0.00000848600$\\
\hline
$(200,0.005)$ & $1.880768$ & $3.373520$ & $0.00003801600$\\
\hline
$(500,0.5)$ &   $5.529987$ & $6.496250$ & $0.00000122200$\\
\hline
\end{longtable}
\end{centering}

\begin{centering}
\begin{longtable}{|c|c|c|c|}
\caption{Empirical complexities for a sample of size $100000$ of truncated Zeta-Dirichlet random variables}\label{trunc_zeta_dir_ky_implementaton_prog_experiment_data}\\ \hline
Par. $u$ & Ent. & Emp. est. $\mathbf{E}(T)$ & Ave. time gen. (ms)\\ \hline
\endfirsthead
Par. $u$ & Ent. & Emp. est. $\mathbf{E}(T)$ & Ave. time gen. (ms)\\ \hline
\endhead
\multicolumn{4}{c}{Continued on next page}
\endfoot
\endlastfoot
\hline
$1/64$ & $7.921181$ & $8.926670$ & $0.000452516$\\
\hline
$1/4$ & $7.281616$ & $8.501400$ & $0.000468623$\\
\hline
$1$ & $5.354125$ & $6.240620$ & $0.000483760$\\
\hline
\end{longtable}
\end{centering}
The Zeta-Dirichlet has unbounded entropy for $0<u\leq 1$. However once we truncate its infinite support to a finite one, the resulting truncated distribution has bounded entropy. The sampling algorithm is applied on the distribution with truncated support.

\subsubsection{The Fast Roller Dice algorithm}\label{sect_lumbroso_fast_dice}

We judge important to recall a result from Lumbroso's PhD thesis \cite{Lum_PHD_2012} concerning the sampling of uniform distribution. Suppose we want for instance to simulate a dice with six faces so that $\mathbf{p}=(\frac{1}{6},\frac{1}{6},\frac{1}{6},\frac{1}{6},\frac{1}{6},\frac{1}{6})$ given that we have an access to an instance of \texttt{RandomBit}. We observe that the binary expansion of $\frac{1}{6}$ is $0.0\overline{01}$ where $\overline{01}$ means $01$ is repeated ad infinitum. The amount of randomness in $\mathbf{p}$ is $H(\mathbf{p})=\log_{2}(6)$ where the latter is the binary entropy of $\mathbf{p}$. Thus for an optimal algorithm, we expect between $\log_{2}(6)$ and $\log_{2}(6)+2$ calls to \texttt{RandomBit} and, from an information theoretical point of view, we cannot do better. Figure \ref{fig:dice} shows the tree with an infinite countable number of lists for the simulation of the dice where the loops must be seen as infinite repetitions of the corresponding subtrees. Actually there is only one kind repeated subtree on figure \ref{fig:dice} which is for the discrete uniform distribution over three elements since $6=2\cdot 3$.

\begin{figure}
\begin{centering}
\includegraphics{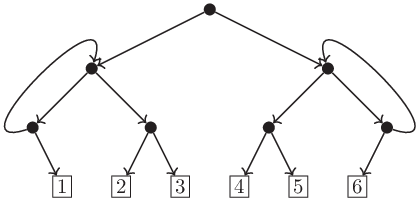}
\caption{\label{fig:dice}A fair dice}
\end{centering}
\end{figure}

With the help figure \ref{fig:dice}, we find $t_j$, $|L_j|$ and $s_j$. We have that $t_1=2$, $|L_1|=0$, and $s_1=2$. For $j\geq 2$, if $j-2\equiv 0\pmod{2}$ then $t_j=4$, $|L_j|=0$, and $s_j=4$. For $j\geq 2$, if $j-2\equiv 1\pmod{2}$ then $t_j=8$, $|L_j|=6$, and $s_j=2$. The Fast Dice Roller from \cite{Lum_PHD_2012} is an efficient implementation of Knuth and Yao ideas for the discrete uniform distribution over $n$ points which is almost identical to algorithm \ref{alg_lum_fast_dice} and exploits the regularity of the quantities $t_j$, $s_j$ and $|L_j|$.

\begin{breakablealgorithm}\label{alg_lum_fast_dice}
\caption{Fast Dice Roller (Lumbroso, 2012)}
\begin{algorithmic}[1]
\raggedright
\Require Integer $n>1$,
\Ensure $X$
\State{$X\leftarrow 0$}
\State{$Y\leftarrow 1$}
\Loop
\State{$Y\leftarrow 2Y$}
\State{$B\leftarrow \texttt{RandomBit}$}
\State{$X\leftarrow 2X+B$}
\If{$Y\geq n$}
\If{$X<n$}
\State{\textbf{Return} $X$}
\Else
\State{$Y\leftarrow Y-n$}\label{Lumb:line:Y:recycle}
\State{$X\leftarrow X-n$}\label{Lumb:line:X:recycle}
\EndIf
\EndIf
\EndLoop
\end{algorithmic}
\end{breakablealgorithm}

We observe that $X$, in the ``loop'' of the Fast Dice Roller, is uniformly distributed. Instructions from lines \ref{Lumb:line:Y:recycle} and \ref{Lumb:line:X:recycle} are executed if and only if $X\geq n$ upon which $X$ is uniformly distributed on $\{n,\ldots,Y-1\}$. Moreover, given that $X\geq n$, the set $\{n,\ldots,Y-1\}\neq \emptyset$ since $Y>X\geq n$ and $\{n,\ldots,Y-1\}$ is translated by $n$ which allows random bits to be ``recycled''.
\begin{theorem}[Lumbroso (2012)]
For all $\alpha>0$, the expected number of calls to \texttt{RandomBit} for the Fast Dice Roller is
\begin{displaymath}
\log_{2}(n)+\frac{1}{2}+\frac{1}{\log 2}-\frac{\gamma}{\log 2}+P(\log_{2}(n))+O(n^{-\alpha})\label{lumbroso_comp},
\end{displaymath}
where $P$ is a trigonometric periodic polynomial and $\gamma$ is the Euler constant.
\end{theorem}

\subsection{DDG tree algorithm and inversion}\label{sect_HH}

We recall the algorithm from Han and Hoshi \cite{HanHos_1997} which is the inversion method for discrete distributions. Given $I\subset\mathbb{N}$, and a probability mass vector $\mathbf{p}=(p_i)_{i\in I}$, the algorithm partitions the interval $[0,1]$ into a countable collection of disjoint subintervals $[q_{i-1},q_{i})$ with $q_{0} = 0$ and $q_{i} = \sum_{k=1}^{i}{p_{k}}$ for $i\in I$. The algorithm refines iteratively a random interval by halving a subset $J\subset [0,1)$ and stops when $J\subset[q_{i-1},q_{i})$ for $i>0$. When $J$ is just small enough such that $J\subset [q_{i-1},q_{i})$, then the outcome $i$ is output. By the probability integral transform, if $U$ is a uniformly distributed random variable on $[0,1]$, then there is unique $i\in I$ such that $q_{i-1}\leq U<q_{i}$. For a binary random source of unbiased i.i.d.\ bits, their algorithm is as follow:

\begin{breakablealgorithm}
\caption{Algorithm from Han and Hoshi}\label{algoHH_algo}
\begin{algorithmic}[1]
\State $T\leftarrow 0$
\State $\alpha_{T}\leftarrow 0$
\State $\beta_{T}\leftarrow 1$
\Repeat
\State $T\leftarrow T+1$
\State $B\leftarrow\texttt{RandomBit}$
\State $\alpha_{T}\leftarrow\alpha_{T-1}+(\beta_{T-1}-\alpha_{T-1})(B\slash2)$
\State $\beta_{T}\leftarrow\alpha_{T-1}+(\beta_{T-1}-\alpha_{T-1})((B+1)\slash2)$
\State $J\leftarrow [\alpha_{T},\beta_{T})$
\Until{$J\subset [q_{i-1},q_{i})$}
\State \textbf{Return} $i$.
\end{algorithmic}
\end{breakablealgorithm}

Let $T$ be the number of random coins needed by \texttt{RandomBit} which is also the number of iterations for the ``repeat'' loop. For $T\geq 1$, $\big[\alpha_{T},\beta_{T}\big)\supset \big[\alpha_{T+1},\beta_{T+1}\big)$. To every node (internal or external) corresponds an interval $\big[\alpha_{T},\beta_{T}\big)$. The root corresponding to the interval $[0,1)$. For each internal node corresponds an interval $\big[\alpha_{T},\beta_{T}\big)$ that is not contained in one of the interval $\big[q_{i-1},q_{i}\big)$, and, if the source produces $B=0$, then the left child corresponds to the interval $\big[\alpha_{T},(\alpha_{T}+\beta_{T})\slash 2)\big)=\big[\alpha_{T+1},\beta_{T+1}\big)$ and, if $B=1$, then the right child corresponds to $\big[(\alpha_{T}+\beta_{T})\slash 2),\beta_{T}\big)=\big[\alpha_{T+1},\beta_{T+1}\big)$. Each leaf (external node) corresponds to an interval $\big[\alpha_{T},\beta_{T}\big)$ entirely contained in $\big[q_{i-1},q_{i}\big)$ upon which the integer $i$ is returned with probability $p_{i}$.

Figures \ref{exec_fig_hanhoshi1} and \ref{exec_fig_hanhoshi2} that are examples of DDG trees for the Han and Hoshi algorithm on some distributions. We observe from figures \ref{exec_fig_hanhoshi1} and \ref{exec_fig_hanhoshi2} that an outcome may appear twice on a level; we can show that it cannot appear also more than twice. In comparison to Knuth and Yao algorithm, an outcome cannot appear more than once on any given levels.

\begin{figure}
\begin{centering}
\includegraphics[scale=0.875]{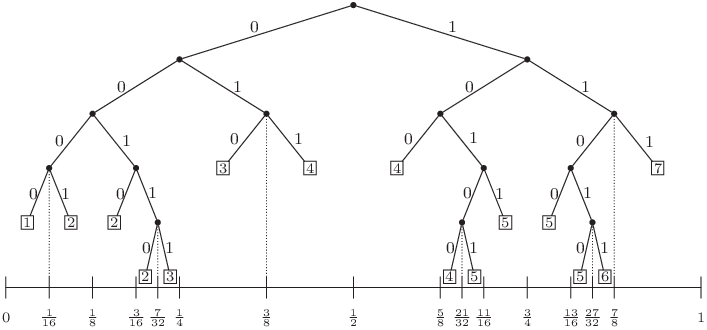}
\caption{Illustration of the algorithm of Han and Hoshi on the vector $(p_1,p_2,p_3,p_4,p_5,p_6,p_7)=\big(\frac{1}{16},\frac{5}{32},\frac{5}{32},\frac{9}{32},\frac{3}{16},\frac{1}{32},\frac{1}{8}\big)$. The cumulative values are $q_1=\frac{2}{32}=(0.00010)_2$, $q_2=\frac{7}{32}=(0.00111)_2$, $q_3=\frac{12}{32}=(0.01100)_2$, $q_4=\frac{21}{32}=(0.10101)_2$, $q_5=\frac{27}{32}=(0.11011)_2$, $q_6=\frac{28}{32}=(0.11100)_2$, and $q_7=\frac{32}{32}=(1.00000)_2$.}\label{exec_fig_hanhoshi1}
\end{centering}
\end{figure}

\begin{figure}
\begin{centering}
\includegraphics[scale=0.875]{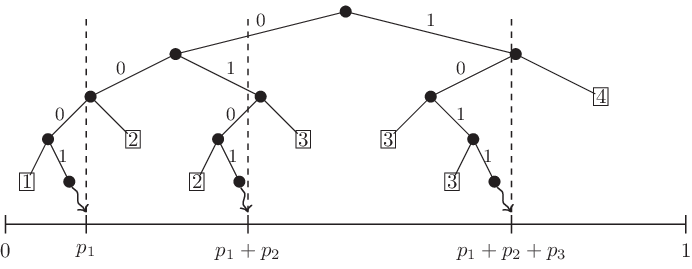}
\caption{Illustration of the Han and Hoshi algorithm on the vector $(p_1,p_2,p_3,p_4)$ such that $p_1=0.0001\cdots$, $p_1+p_2=0.0101\cdots$, and $p_1+p_2+p_3=0.1011\cdots$.}\label{exec_fig_hanhoshi2}
\end{centering}
\end{figure}

\begin{theorem}[Han and Hoshi \cite{HanHos_1997}]\label{HHthm}
The expected number of calls to \texttt{RandomBit} in algorithm \ref{algoHH_algo} is bounded from below by $H(\mathbf{p})$ and from above by $H(\mathbf{p})+3$.
\end{theorem}

\begin{proof}[Proof of theorem \ref{HHthm}]
Our new proof partitions the leaves $L_i$ for symbol $i$ in the DDG tree arbitrarily into two sets, $A_i$ and $B_i$, such that $A_i$ and $B_i$ each possesses at most one leaf at depth $i$. Let $\alpha_i=\sum_{u\in A_i}{2^{-\texttt{d}(u)}}$, $\beta_i=\sum_{u\in B_i}{2^{-\texttt{d}(u)}}$ where $\texttt{d}(u)$ the depth of leaf $u$, so that $p_i=\alpha_i+\beta_i$. By using elementary calculations and facts, we have
\begin{align*}
&\sum_{i=1}^{\infty}{p_i\log_{2}\bigg(\frac{1}{p_i}\bigg)}\leq \sum_{i=1}^{\infty}{\alpha_i\log_{2}\bigg(\frac{1}{\alpha_i}\bigg)}+\sum_{i=1}^{\infty}{\beta_i\log_{2}\bigg(\frac{1}{\beta_i}\bigg)}\leq \sum_{i=1}^{\infty}{p_i\log_{2}\bigg(\frac{1}{p_i}\bigg)}+1.
\end{align*}
Let $(\alpha_{i})_{j}$ be the $j$-th bit in the binary expansion of $\alpha_i$, and let $(\beta_{i})_{j}$ be the $j$-th bit for $\beta_i$. Then we have
\begin{displaymath}
\mathbf{E}(T)=\sum_{j=1}^{\infty}{\frac{j(\alpha_{i})_{j}}{2^{j}}}+\sum_{j=1}^{\infty}{\frac{j(\beta_{i})_{j}}{2^{j}}}= \textrm{I}+\textrm{II}.
\end{displaymath}
As in the proof of theorem \ref{KYthm}, we have
\begin{align*}
\sum_{i=1}^{\infty}{\alpha_i\log_{2}\bigg(\frac{1}{\alpha_i}\bigg)}&\leq \textrm{I}\leq \sum_{i=1}^{\infty}{\alpha_i\log_{2}\bigg(\frac{1}{\alpha_i}\bigg)}+2\alpha_i,\\
\sum_{i=1}^{\infty}{\beta_i\log_{2}\bigg(\frac{1}{\beta_i}\bigg)}&\leq \textrm{II} \leq \sum_{i=1}^{\infty}{\beta_i\log_{2}\bigg(\frac{1}{\beta_i}\bigg)}+2\beta_i,
\end{align*}
and hence, using the above,
\begin{displaymath}
\sum_{i=1}^{\infty}{p_i\log_{2}\bigg(\frac{1}{p_i}\bigg)}\leq \mathbf{E}(T)\leq\sum_{i=1}^{\infty}{p_i\log_{2}\bigg(\frac{1}{p_i}\bigg)}+1+2\sum_{i=1}^{\infty}{p_i}\leq H(\mathbf{p})+3.
\end{displaymath}
\end{proof}

\subsection{Sequential generation of random variables}\label{sect_batch_gen}

In this section, we explore the generation of a sample of i.i.d.\ random variables distributed according to some discrete probability distribution. The term batch generation is sometimes used in the literature as well. A batch is a sample of i.i.d.\ random variables generated sequentially using a sampling method. Before studying batch generation with optimal asymptotic complexity in section \ref{sect_assymptotic_batch}, and its core component which is randomness extraction in section \ref{sect_rnd_ext}, we explain key concepts for general DDG-tree based algorithms in section \ref{sect_generic_DDG}.

The main goal of this section is to show that slight modifications of any DDG-tree algorithms, used within a batch generation algorithm, allows to reach the optimal expected complexity. More precisely, suppose $\mathbf{p}$ is a probability vector for some random variable $Y$ and that a sample of size $n$ instances of $Y$ is generated sequentially; then the expected number of unbiased random bits needed to generate the batch is tightly concentrated around $nH(Y)$ for sufficiently large values of $n$. Here $H(Y)$ denotes the binary entropy of the distribution of $Y$. We know from the previous sections that upon the generation of a single instance of $Y$ using for example one of the DDG-tree based algorithms, the expected complexity of the method lies between $H(Y)$ and $H(Y)+O(1)$.

Two modifications are to be done at a DDG-tree algorithm to allow asymptotic optimality in batch generation. The first consists to return the depth of a leave upon halting together with the label (outcome) of that leave. By returning a pair (depth, label), we can extract random bits sequentially so that as $n$ gets larger and larger the interval for the expected complexity shrinks around $nH(Y)$. The second modification is the use of $\mathtt{FetchBit}$ as an algorithm to either retrieve previously recycled bits from a queue or call $\mathtt{RandomBit}$ when the queue is empty.

For clarity, we split this section \ref{sect_batch_gen} into three sub-sections. We decided to insert section \ref{sect_generic_DDG} on generic DDG-tree based algorithm within section \ref{sect_batch_gen} on batch generation to keep a logical flow logical, but it could deserve an entire section.

\subsubsection{Generic DDG-tree based algorithms}\label{sect_generic_DDG}

Suppose we aim to generate an outcome of a random variable $Y$ with probability vector $\mathbf{p}=(p_i)_{i\in I}$ for $I\subset\mathbb{Z}$. We explain here what a generic DDG-tree algorithm is and why Knuth and Yao or Han and Hoshi algorithms are special nearly optimal algorithms. Let $\texttt{A}$ be a DDG-tree algorithm to sample $\mathbf{p}$. We write the binary expansions of the $p_{i}$'s as
\begin{equation}
\mathbf{P}\{Y=i\}=p_{i}=\sum_{j>0}{\frac{p_{ij}}{2^{j}}}\quad\text{for $p_{ij}\in\{0,1\}$}.\label{eq_dist_target}
\end{equation}

Let $\Lambda$ be the set of leaves of the DDG-tree. We write $\texttt{label}(u)$ to denote the label of a leaf $u\in \Lambda$. Labels are instances of $Y$ which are distributed according to $\mathbf{p}$. We write $\texttt{depth}(u)$ for the depth of a leaf $u\in \Lambda$. A useful variant of the traditional DDG-tree based method is one which returns a random pair $(X,Y)=(\texttt{depth}(u),\texttt{label}(u))$.

We now discuss some facts of that variant, and, for that, it is convenient to define the matrix $\Theta$ with integer entries denoted by $\Theta_{ij}\geq 0$ for $i\in I$ and $j>0$, by
\begin{displaymath}
\Theta_{ij}=\text{card}\big\{u\in\Lambda\colon\texttt{label}(u)=i,\hspace{2pt}\texttt{depth}(u)=j\big\}.
\end{displaymath}
The entry $\Theta_{ij}$ is the number of leaves at depth $j$ with outcomes labelled $i$. Necessarily we have also that $\sum_{i\in I}{\Theta_{ij}} \leq 2^{j}$ for all $j>0$, that is, a sum over rows for a fixed column is bounded. A sum over columns for a fixed row is usually unbounded since its corresponding $p_i$ is an irrational number most often. We observe for a given $i\in I$ that $\max\{\Theta_{ij}\colon j>0\}\geq 1$ because $p_{i}>0$. If the latter maximum is strictly less than one that is equal to zero and since all entries on a given row are less than the maximum, then all entries are zeros for this row which means that the algorithm would never outputs the symbol associated with this given row. Therefore let us define the soon-useful quantities $\kappa_{i}$ by
\begin{align*}
\kappa_{i}&=\max\{\Theta_{ij}\colon j>0\}\quad\text{for $i\in I$, and,}\\
\kappa&=\max\{\Theta_{ij}\colon i\in I,j>0\}=\max\{\kappa_{i}\colon i\in I\}.
\end{align*}
The quantity $\kappa$ is necessarily bounded above otherwise there would be a non-zero probability to do not halt.

For every row of $\Theta$ or equivalently for every outcome from $\mathbf{p}$, we associate a set of leaves in order to perform conditional analysis. More precisely, for $i\in I$, let
\begin{displaymath}
\Lambda_{i}=\{u\in\Lambda\colon\texttt{label}(u)=i\}.
\end{displaymath}
We observe that if a DDG algorithm is correct, then it holds that
\begin{displaymath}
p_{i}=\sum_{u\in\Lambda_{i}}{\frac{1}{2^{\texttt{d}(u)}}}=\sum_{j>0}{\frac{\Theta_{ij}}{2^{j}}}.
\end{displaymath}

\begin{remark}[The case of Knuth and Yao]
Knuth and Yao's algorithm is designed such that $\Theta_{ij}=p_{ij}\leq 1$ for all $i\in I$ and $j>0$. Thus we have $\kappa=1$. For examples, refer to figures \ref{figDDG} and \ref{fig:dice} from section \ref{sect_KY_algo}.
\end{remark}

\begin{remark}[The case of Han and Hoshi]
Han and Hoshi's algorithm is designed such that $\Theta_{ij}\leq 2$ for all $i\in I$ and $j>0$. Thus we have $\kappa=2$. For examples, refer to figures \ref{exec_fig_hanhoshi1} and \ref{exec_fig_hanhoshi2} from section \ref{sect_HH}.
\end{remark}

For each $i \in I$, we denote $J_{i}=\{j>0\colon\Theta_{ij}\neq 0\}$. A useful quantity for later is the binary joint entropy of $X$ and $Y$ given by
\begin{align}
H(X,Y)&=\sum_{i\in I}\sum_{j>0}{\frac{\Theta_{ij}}{2^{j}}\log_{2}\bigg(\frac{2^{j}}{\Theta_{ij}}\bigg)}=\sum_{i\in I}\sum_{j\in J_{i}}{\frac{\Theta_{ij}}{2^{j}}\log_{2}\bigg(\frac{2^{j}}{\Theta_{ij}}\bigg)}\label{eq_joint_ent_var}\\
&\quad\leq \sum_{i\in I}\sum_{j\in J_{i}}{\frac{j\Theta_{ij}}{2^{j}}}=\sum_{i\in I}\sum_{u\in \Lambda_{i}}{\frac{\texttt{d}(u)+1+\delta_{i}}{2^{\texttt{d}(u)}}}=H(Y)+\sum_{i\in I}{p_{i}(1+\delta_{i})}\nonumber\\
&\quad\leq H(Y)+1+\sup\{i\in I\colon\delta_{i}\}.\label{eq_joint_ent0}
\end{align}
For instance in the case of the Knuth and Yao's algorithm, (\ref{eqn:expcomp_foo}) from section \ref{sect_KY_algo} shows that $\sup\{i\in I\colon\delta_{i}\}\leq 1$ in (\ref{eq_joint_ent0}). For Han and Hoshi's algorithm, we can show that $\delta_{i}\leq 2$. A generic DDG-tree based algorithm must have $\delta_{i}\leq\kappa$ for all $i\in I$.

\begin{remark}
In some cases, it could be possible to get the fluctuation term from (\ref{eq_joint_ent_var}) by using complex analysis methods. Here we simply used the fact $1\leq \Theta_{ij}$ for $j\in J_{i}$ and hence $-\log_{2}(\Theta_{ij})\leq 0$.
\end{remark}

To prepare us for sections \ref{sect_rnd_ext} and \ref{sect_assymptotic_batch}, we now observe that, conditional on the event $\{Y=i\}$, the amount of randomness left in the sampling process is given by
\begin{align}
&H(X\mid Y)=\sum_{i\in I}{p_{i}H(X\mid Y=i)}=H(X,Y)-H(Y)\leq 1+ \sup\{i\in I\colon\delta_{i}\}\label{eq_cond_poutine}.
\end{align}
If a sample of independent pairs $\{(X_i,Y_i)\}_{i=1}^{n}$ is generated sequentially, therefore we could at most extract $1+\sup\{i\in I\colon\delta_{i}\}$ bits every time we generate a pair and use the extracted bits for the next pair. The extraction of randomness allows us to reach asymptotic optimal complexity in probability for any DDG-tree based algorithms. We emphasize that the coupling $(X,Y)$ is correlated. The marginal variable $X$ is the number of random bits required to generate $Y$. Conditional on the event $\{Y=i\}$, we have that
\begin{align}
X&=\log_{2}\bigg(\frac{1}{p_{i}}\bigg)+\sum_{j>0}{\mathbf{P}\{X=j\mid Y=i\}\log_{2}\bigg(\frac{1}{\mathbf{P}\{X=j\mid Y=i\}}\bigg)}\nonumber\\
&=\log_{2}\bigg(\frac{1}{p_{i}}\bigg)+H(X\mid Y=i),\nonumber\\
\mathbf{E}(X)&=\mathbf{E}_{Y}\big(\mathbf{E}(X\mid Y)\big)=\sum_{i\in I}{p_{i}\bigg(\log_{2}\bigg(\frac{1}{p_{i}}\bigg)+H(X\mid Y=i)\bigg)}\\
&=H(Y)+H(X\mid Y)\label{average_comp_1out}.
\end{align}

We make before ending this section an observation about the rate at which $\mathbf{P}\{X=x,Y=y\}$ decays when $y$ is kept fixed and $x$ increased. Indeed, we have
\begin{equation}
\mathbf{P}\{X=x,Y=y\}\leq \frac{1}{2}\mathbf{P}\{X=x-1,Y=y\}\quad\text{for all $x\geq 0$ and $y\in I$.}\label{eq_crapeau_qui_saute}
\end{equation}
Inequality (\ref{eq_crapeau_qui_saute}) means that the probability of halting decreasing geometrically for every new call to \texttt{RandomBit} and for any $y\in I$. Equivalently, inequality (\ref{eq_crapeau_qui_saute}) can be stated as
\begin{equation}
\mathbf{P}\{X=x\mid Y=y\}\leq \frac{1}{2}\mathbf{P}\{X=x-1\mid Y=y\}\quad\text{for all $x\geq 0$ and $y\in I$.}\label{eq_crapeau_qui_saute2}
\end{equation}

\subsubsection{Randomness extraction}\label{sect_rnd_ext}

In this section we discuss randomness extraction which can be seen to some extents as the inverse of random variate generation. Randomness extraction turns a sequence of i.i.d. random variables into a sequence of i.i.d. random unbiased Bernoulli random variables.

For intuitional purposes, we mention ahead of time the main idea to keep in mind for this section. Suppose a DDG-tree algorithm returns a random pair $(X,Y)$ as explained in section \ref{sect_generic_DDG} where $X$ is the depth of a leaf with label $Y$ when halting. Given a sample of size $n$ pairs, we can wonder what is the conditional likelihood of the depths given the labels. The conditional event $\{X\mid Y=i\}$ pertains to the number $X$ of consumed bits given label $Y=i$. We show here that the conditional likelihood for large enough $n$ is close to the conditional entropy $H(X\mid Y)$. If we denote the outcomes $(x_{\ell},y_{\ell})$ for $1\leq \ell\leq n$, the conditional likelihood $\prod_{\ell=1}^{n}{\mathbf{P}\{X=x_{\ell}\hspace{2pt}|\hspace{2pt}Y=y_{\ell}\}}$ suitably normalized is highly concentrated around $H(X\mid Y)$ as in (\ref{eq_cond_poutine}) for large $n$ by using the weak law of large numbers as explained for instance in Kullback \cite{Kul_book_1997ed}.

\begin{theorem}\label{quibivilortrinuloxifene}
With the notation of before, there exits an algorithm (described below) such that with inputs $(X_1,Y_1),(X_2,Y_2)\ldots,(X_n,Y_n)$ outputs a sequence of i.i.d.\ unbiased bits of random length $R_n$ where
\begin{displaymath}
\frac{R_n}{n}\stackrel{\mathrm{p}}{=}H(Y\mid X)\text{ as }n\rightarrow \infty.
\end{displaymath}
\end{theorem}

The next algorithm maps a uniform random variable $U=0.U_1U_2\cdots$ to a random sequence $(X_1,Y_1),\ldots{},(X_n,Y_n)$.\ The bits $U_i$'s are unbiased i.i.d.\ and independent of $X_1,X_2,\ldots,X_{n}$.
\begin{breakablealgorithm}\label{algo_extractx}
\caption{Randomness extraction adapted to DDG-tree based algorithms}
\begin{algorithmic}[1]
\raggedright
\Require A sequence of pairs $(x_1,y_1)$, \ldots{}, $(x_n,y_n)$ as previously described.
\Ensure A sequence of unbiased i.i.d. bits.
\State $U_{0}^{-}\leftarrow 0$
\State $U_{0}^{+}\leftarrow 1$
\For{$\ell = 1$ \textbf{to} $n$}
\State $U_{\ell}^{-}\leftarrow U_{\ell-1}^{-}+\big(U_{\ell-1}^{+}-U_{\ell-1}^{-}\big)\mathbf{P}\{X=x_{\ell}-1\mid Y=y_{\ell}\}$
\State $U_{\ell}^{+}\leftarrow U_{\ell-1}^{-}+\big(U_{\ell-1}^{+}-U_{\ell-1}^{-}\big)\mathbf{P}\{X=x_{\ell}\mid Y=y_{\ell}\}$
\EndFor
\State $R_{n}\leftarrow \max\big\{t\geq 0\phantom{1}:\phantom{1}\lfloor 2^{t}U_{n}^{-}\rfloor = \lfloor 2^{t}U_{n}^{+}\rfloor\big\}$ \Comment{$R_n$ is the number of bits of the longest prefix common to both $U_{n}^{-}$ and $U_{n}^{+}$.}
\State{\textbf{Return} $\lfloor 2^{R_n}U_{n}^{-}\rfloor=U_{1}^{-1}\cdots U_{R_{n}}^{-1}$}
\end{algorithmic}
\end{breakablealgorithm}

For the correctness of algorithm \ref{algo_extractx}, we observe that the intervals $[U_{\ell}^{-},U_{\ell}^{+}]$ are nested by inequality (\ref{eq_crapeau_qui_saute}) from section \ref{sect_generic_DDG}, that is
$[U_{\ell}^{-},U_{\ell}^{+}] \supseteq [U_{\ell+1}^{-},U_{\ell+1}^{+}]$ for $\ell>0$. By the nested property, $U=\limsup_{n\to\infty}{U_{n}^{-}} = \liminf_{n\to\infty}{U_{n}^{+}}$, and we have that
\begin{displaymath}
U\in [U_{n}^{-},U_{n}^{+}]=\bigcap_{\ell=1}^{n}{[U_{\ell}^{-},U_{\ell}^{+}]}.
\end{displaymath}
For every $\ell$-th iteration, we have $U(U_{\ell}^{+}-U_{\ell}^{-})+U_{\ell}^{-}$ is uniformly distributed on the interval $[U_{j}^{-},U_{j}^{+}]$. Because $R_n=\max\big\{t\geq 0\phantom{1}:\phantom{1}\lfloor 2^{t}U_{n}^{-}\rfloor = \lfloor 2^{t}U_{n}^{+}\rfloor\big\}$, then it holds that
\begin{displaymath}
\frac{1}{2^{R_n}}\lfloor 2^{R_n}U\rfloor\leq U_{n}^{-}\leq U\leq U_{n}^{+}\leq \frac{1}{2^{R_n}}\Big(\lfloor2^{R_n}U\rfloor+1\Big),
\end{displaymath}
and thus the bits $U_{1},U_2\ldots,U_{R_n}$ are unbiased i.i.d.

In order to keep the notation compact, we write $\mu_{ij}=\mathbf{P}\{X=j,Y=i\}$, $\mu_{i}=p_i=\mathbf{P}\{Y=i\}$, and $\mu_{j}=\mathbf{P}\{X=j\}$; it is clear to which of the marginal distributions or the joint we refer to in the next proof.

\begin{proof}[Proof of theorem \ref{quibivilortrinuloxifene}]
Let $t\in\mathbb{R}$ and consider the two cases $\{R_n\geq t\}$ and $\{R_n<t\}$. We show that $t$ is concentrated around $n\sum_{i\in I}{\mu_iH(X\mid Y=i)}$. More precisely, we show for all $\epsilon>0$ that
\begin{displaymath}
\lim_{n\to\infty}\mathbf{P}\bigg\{\bigg|\frac{R_n}{n}-\mathbf{E}\bigg(-\log_{2}\frac{\mu_{XY}}{\mu_{Y}}\bigg)\bigg|<\epsilon\bigg\}=0.
\end{displaymath}

For all $t>0$, we have that $\big\{R_n\geq t\big\}\subseteq \big\{U_{n}^{+}-U_{n}^{-}<2^{-t}\big\}$ and that
\begin{equation}
\Big\{U_{n}^{+}-U_{n}^{-}<\frac{1}{2^{t}}\Big\}=\bigg\{\prod_{\ell=1}^{n}{\frac{\mu_{X_{\ell}Y_{\ell}}}{\mu_{Y_{\ell}}}<\frac{1}{2^{t}}}\bigg\}=\bigg\{-\sum_{\ell=1}^{n}\log_{2}\bigg(\frac{\mu_{X_{\ell}Y_{\ell}}}{\mu_{Y_{\ell}}}\bigg)>t\bigg\}\label{mimayoxilugenufene}
\end{equation}
To apply the weak law of large numbers on (\ref{mimayoxilugenufene}), choose an $\epsilon>0$ and set
\begin{displaymath}
t=n\bigg(\mathbf{E}\bigg(-\log_{2}\frac{\mu_{XY}}{\mu_{X}}\bigg)+\epsilon\bigg)=n\bigg(\sum_{i\in I}{p_{i}H(X\mid Y=i)}+\epsilon\bigg)=n(H(X\mid Y)+\epsilon).
\end{displaymath}
Then we obtain
\begin{align*}
&\lim_{n\to\infty}\mathbf{P}\bigg\{\frac{R_n}{n}\geq \mathbf{E}\bigg(-\log_{2}\frac{\mu_{XY}}{\mu_{Y}}\bigg) + \epsilon\bigg\}\\
&\quad\leq\lim_{n\to\infty}\mathbf{P}\bigg\{\bigg(-\frac{1}{n}\sum_{\ell=1}^{n}\log_{2}\bigg(\frac{\mu_{X_{\ell}Y_{\ell}}}{\mu_{Y_{\ell}}}\bigg)\bigg)-\mathbf{E}\bigg(-\log_{2}\frac{\mu_{XY}}{\mu_{Y}}\bigg)>\epsilon\bigg\}=0.
\end{align*}
For the event $\{R_n<t\}$, we have that
\begin{align*}
\big\{R_n< t\big\}&\subseteq \big\{\lfloor 2^{t}U_{n}^{+}\rfloor > \lfloor 2^{t}U_{n}^{-}\rfloor \big\}\subseteq\Big\{U_{n}^{+}-U_{n}^{-}\geq \frac{1}{2^{t}}\Big\}.
\end{align*}
By a same argument as for the event $\{R_n>t\}$, we choose $\epsilon>0$, set $t=n(H(X\mid Y)-\epsilon)$, and evaluate limits to obtain that
\begin{align*}
\lim_{n\to\infty}\mathbf{P}\bigg\{\frac{R_n}{n}\leq \mathbf{E}\bigg(-\log_{2}\frac{\mu_{XY}}{\mu_{Y}}\bigg) - \epsilon\bigg\}=0.
\end{align*}
\end{proof}

Now that we have showed theorem \ref{quibivilortrinuloxifene}, the next section explains how to recycle bits sequentially.

\subsubsection{Asymptotic complexity for batch generation}\label{sect_assymptotic_batch}

Our method of batch generation is valid for any DDG-tree based algorithm such that $H(Y)<\infty$. Given a global queue data structure, denoted by $Q$, that contains random bits, our algorithm fetches bits from a non-empty $Q$ through an operation that we denote \texttt{FetchBit}. If $Q$ is empty, then \texttt{FetchBit} invokes \texttt{RandomBit}. The operation \texttt{FetchBit} drives the DDG-tree algorithm and that the analysis following algorithm \ref{alg_batch} is mostly about the expected size of $Q$.

\begin{breakablealgorithm}
\caption{Batch generation}\label{alg_batch}
\begin{algorithmic}[1]
\Require $Q$, $\mathbf{p}$, and $n$
\Ensure Batch $(Y_1,\ldots,Y_n)$
\State{$Q\leftarrow \emptyset$ \Comment{Initially, the queue is empty.}}
\State{$R_{0}\leftarrow 0$ \Comment{There is no ``recycled'' bit initially.}}
\For{$i=1$ \textbf{to} $n$}
\State Obtain $(X_i,Y_i)$ by a DDG-tree algorithm. \Comment{The DDG algorithm uses the operation \texttt{FetchBit} to get bits either from the source or from the queue $Q$.}
\State{\textbf{Return} $Y_i$}
\State{Feed $(X_i,Y_i)$ to the retrieval algorithm (randomness extraction procedure), and recover $R_i-R_{i-1}$ bits which are added to $Q$.}
\EndFor
\end{algorithmic}
\end{breakablealgorithm}

\begin{theorem}\label{thm_batch_algo}
Whenever $H(Y)$ is bounded, algorithm \ref{alg_batch} uses $T_n$ random bits where
\begin{displaymath}
\frac{T_{n}}{n}\stackrel{\textrm{p}}{\rightarrow} H(Y)\quad\text{as $n\rightarrow \infty$.}
\end{displaymath}
\end{theorem}

\begin{remark}
By Shannon's lower bound, the procedure is asymptotically optimal for we have $\mathbf{E}(T_n)\geq nH(Y)$.
\end{remark}

\begin{proof}[Proof of theorem \ref{thm_batch_algo}]
We choose a large integer constant $k$ and look at $T_{nk}$. Let $Q_{t}$ be the size of the queue at time $t$, and set $Q_{0}=0$. For $j\in\{1,\ldots, nk\}$, let $X_{j}$ be the number of bits needed to generate $Y_{j}$ without extraction. By (\ref{average_comp_1out}), we have that $\mathbf{E}(X_{j})=H(Y)+H(X\mid Y)$. The random variables $X_j$ are i.i.d.\ Then we have the following simple identity:
\begin{displaymath}
T_{nk}=\sum_{j=1}^{nk}{X_{j}}-R_{nk}+Q_{nk}.
\end{displaymath}
Given $k>0$, we have by the law of large numbers and by theorem \ref{quibivilortrinuloxifene} from section \ref{sect_rnd_ext}, we have that
\begin{displaymath}
\frac{X_1+X_2+\ldots+X_{nk}}{nk}\stackrel{\mathrm{p}}{\rightarrow}\mathbf{E}(X)=H(Y)+H(X\mid Y)\text{ as $n\rightarrow \infty$.}
\end{displaymath}
Thus it follows for some $\delta_{nk}$ that
\begin{align*}
\frac{T_{nk}}{nk}&=H(Y)+H(X\mid Y)-H(X\mid Y)+\delta_{nk}+\frac{Q_{nk}}{nk}\\
&=H(Y)+\delta_{nk}+\frac{Q_{nk}}{nk}\quad\text{with $\delta_{nk}\stackrel{\mathrm{p}}{\rightarrow} 0$ as $n\to\infty$.}
\end{align*}
If $\big(Q_{nk}\slash nk\big)\stackrel{\textrm{p}}{\rightarrow}0$ as $n\rightarrow \infty$, then the result follows. For this, we need only to consider an upper bound, since $Q_{nk}\geq 0$, and then
\begin{align*}
Q_{nj}&\leq Q_{nj-1}+\big(R_{nj}-R_{n(j-1)}\big)-\min_{1\leq j\leq k}\{T_{nj}-T_{n(j-1)}, Q_{n(j-1)}\}.
\end{align*}
Since $\big(R_n\slash n\big)\stackrel{\textrm{p}}{\rightarrow} H(X\mid Y)$, we have
\begin{align}
\max_{1\leq j\leq k}&\bigg|\frac{R_{nj}-R_{n(j-1)}}{n}-H(X\mid Y)\bigg|\stackrel{\textrm{p}}{\rightarrow}0\label{grumblegrumble1}
\end{align}
and
\begin{align}
\max_{1\leq j\leq k}&\bigg|\frac{T_{nj}-T_{n(j-1)}}{n}-H(Y)\bigg|\stackrel{\textrm{p}}{\rightarrow}0.\label{grumblegrumble2}
\end{align}
Fix $\epsilon>0$, and let $A$ be the event that both left-hand sides in (\ref{grumblegrumble1}) and (\ref{grumblegrumble2}), respectively, are less than $\epsilon$, so that $\mathbf{P}\{A^{c}\}=o(1)$. The critical observation is that on $A$,
\begin{align*}
Q_{nj}&\leq \left\{
\begin{array}{ll}
Q_{n(j-1)} + \big(H(X\mid Y)-H(Y)\big)n\quad\text{if  $Q_{n(j-1)} \geq \big(H(Y)-\epsilon\big)n$,}\\
Q_{n(j-1)}+  \big(H(X\mid Y)+\epsilon\big)n- Q_{n(j-1)}\quad\text{else.}
\end{array} \right.\\
&\leq  \left.
\begin{array}{ll}
\max\big\{Q_{n(j-1)},\hspace{2pt}\big(H(X\mid Y)+\epsilon\big)\big\}\quad\text{if $2\epsilon\leq H(Y)$},
\end{array}\right.
\end{align*}
and therefore
\begin{displaymath}
\max_{1\leq j\leq k}\{Q_{nj}\}\leq \big(H(X\mid Y)+\epsilon)n\quad\text{and}\quad\frac{Q_{nk}}{nk}\leq \frac{H(X\mid Y)+\epsilon}{k}.
\end{displaymath}
If we choose $k$ large enough such that \mbox{$\big(\big(H(X\mid Y)+\epsilon\big)\slash k\big)\leq \epsilon$}, then
\begin{equation*}
\mathbf{P}\bigg\{\frac{Q_{nk}}{nk}>\epsilon\bigg\}\leq \mathbf{P}\{A^{c}\}=o(1).
\end{equation*}
\end{proof}

\section{Continuous variate generation}\label{sect_main_sect_cts}

This section is divided into various parts that reflect the main approaches to random variate generation. However the most important section is our plea in section \ref{sect_cts_gen} for the natural choice of the Wasserstein metric to measure distances between outputs and their corresponding ideal distributions. In section \ref{sect_wasserstein_inversion}, we establish an important result in the one-dimensional case that connects the Wasserstein $L_{\infty}$ metric to the inverses of the distribution functions for which we measure the distance between. Many facts from sections \ref{sect_cts_gen} and \ref{sect_wasserstein_inversion} apply to any continuous distributions, singular or absolutely continuous. However, if we impose that the target distribution is absolutely continuous, then, we obtain our main result, mentioned in the introduction, which describes the expected complexity of the number of random bits to generate an $\epsilon$-accurate random outcome for $\epsilon>0$. Our result states that $-\log_{2}(\epsilon)+H(f)+O(1)$ random bits are expected whenever $H(f)<\infty$, and our complexity result is reminiscent to those for the discrete cases involving the entropy. Given the importance of the differential entropy, we make a little interlude in section \ref{sect_diff_ent_usage} to explain when the use of the differential entropy is justified in the context of sampling an absolutely continuous distribution.

We discuss various methods to generate absolutely continuous random variables and their upper bounds in sections \ref{sect_discretization}, \ref{sect_inversion}, \ref{sect_bisection}, \ref{sect_vN} and \ref{sect_convolution} which analyzes the discretization, inversion, bisection, Von Neuman and convolution methods, respectively. Those methods, excepted for the convolutional approach in \ref{sect_convolution}, concern absolutely continuous distributions.

\subsection{Wasserstein metric and universal lower bound}\label{sect_cts_gen}

From this section and the following ones, we deal with continuous probability distributions. Results from the previous sections settle the discrete case and we are now ready to extend the former sections to the continuous case. Let us recall the Wasserstein metric which is central in our extension to the continuous case. In the continuous, we must deal with accuracy necessarily and the Wasserstein $L_\infty$-metric takes into account the accuracy in our computations. Generally speaking let $X$ and $Y$ be two random variables. Let $F$ and $G$ be the distributions of $X$ and $Y$, respectively. If $\mathcal{M}$ denote the product space of all distributions of pairs \mbox{$(X,Y) \in \mathbb{R}^{d}\times\mathbb{R}^{d}$} with fixed marginal distributions $F$ and $G$, respectively, then the Wasserstein $L_\infty$-distance between $X$ and $Y$, or between $F$ and $G$, is
\begin{displaymath}
\begin{split}
W_{p}(F,G)=\inf\big\{\esssup\|X-Y\|_{p}\hspace{2pt}:\hspace{2pt}(F,G)\in \mathcal{M}\},
\end{split}
\end{displaymath}
where $\esssup$ denotes the essential supremum. The Wasserstein $L_\infty$-metric defines intrinsically a distance between $X$ and $Y$ that is $\mathrm{dist}_{p}(X,Y)= W_{p}(F,G)$. By definition, if $\mathrm{dist}_{p}(X,Y) \le \epsilon$, then there exists a random variable $Y$ (output) coupled with $X$ (target) such that \mbox{$\esssup\|X-Y\|_{p} \le \epsilon$} that is, with probability one,
 $\|X-Y\|_{p}<\epsilon$. Clearly we observe that the three axioms for a distance metric are satisfied.

Before we prove our main theorem that is mentioned in the introduction, we recall definition \ref{def_arbsamalg}, introduce a new one, and also proves useful lemmas. Definition \ref{def_arbsamalg} defines an $(\epsilon,p)$-sampling algorithm $\mathtt{A}$ as a probabilistic algorithm such that for a target random variable $X\in\mathbb{R}^{d}$, algorithm $\mathtt{A}$ returns an output random variable $Y\in\mathbb{R}^{d}$ such that $\|X-Y\|_{p}<\epsilon$ with probability one, and $\mathtt{A}$ halts when \texttt{RandomBit} gets invoked $T$ times. In the following definition, $G$ denotes a graph, $V$ its corresponding set of vertices and $E$ its corresponding set of edges.
\begin{definition}[sampling graph]\label{defn_samplinggraph}
Given the joint random variable $(X,Y)\in\mathbb{R}^{d}\times\mathbb{R}^{d}$ of an $(\epsilon,p)$-sampling algorithm for the target $X$ with output $Y$ and a countable partition $\mathcal{A}$ of $\mathbb{R}^{d}$, we say that $G=(V,E)$ is a sampling graph for $X$ on partition $\mathcal{A}$ if
\begin{enumerate}
\item[(1)] $V=\mathcal{A}$,
\item[(2)] $(A,B)\in E\Longleftrightarrow \inf\big\{\|x-y\|_{p}\colon (x,y)\in V\times V\big\}<\epsilon$.
\end{enumerate}
\end{definition}
We denote by $\Delta$ be the maximal degree of any vertex of $G$. We observe that the partition $\mathcal{A}$ induces a discrete distribution with probability masses given by $\mathbf{P}\{X\in A\}$ for $A\in\mathcal{A}$; we denote by $H(X\mid\mathcal{A})$ the entropy of the latter discrete distribution, that is,
\begin{displaymath}
H(X\mid\mathcal{A})=-\sum_{A\in\mathcal{A}}{\mathbf{P}\{X\in\mathcal{A}\}\log_{2}\mathbf{P}\{X\in\mathcal{A}\}}.
\end{displaymath}
For $B\in \mathcal{A}$, we denote the neighbourhood of $B$, with respect to $\epsilon$ and $p$, by $\partial(B)=\{A\in A\colon (A,B)\in E\}$.
\begin{lemma}\label{lowerbndlem}
With the above notation, if $T$ is the number of bits used by an $(\epsilon,p)$-sampling algorithm for a target $X$ and an output $Y$, then
\begin{displaymath}
\mathbf{E}(T)\geq \sup_{\mathcal{A}}\big\{H(X\mid\mathcal{A})-\log_{2}(\Delta+1)\big\}.
\end{displaymath}
\end{lemma}

\begin{proof}[Proof of lemma \ref{lowerbndlem}]
Let $X$ and $Y$ be two (dependent) random variables on $\mathbb{R}^{d}$ with $X$ and $Y$ the target and output of an $(\epsilon,p)$-sampling algorithm. For notational purposes, let us denote the joint distribution of $\{X\in A,Y\in B\}$ by $\mu_{AB}=\mathbf{P}\{X\in A,Y\in B\}$ and the marginal distributions by $\mu_{A}=\mathbf{P}\{X\in A\}$ or $\mu_{B}=\mathbf{P}\{Y\in B\}$; in the case of the marginal distributions, it is always clear to which one we refer to. In the inequality that follows, we use the fact $\mu_{AB}=0$ if $(A,B)$ is not an edge of $G$ as well as the definition of convexity. We then have
\begin{align}
&\big|H(X\mid\mathcal{A})-H(Y\mid\mathcal{A})\big|=\Bigg|\sum_{A\in\mathcal{A}}{\mu_{A}\log_{2}\Bigg(\frac{1}{\mu_{A}}\Bigg)}-\sum_{B\in\mathcal{A}}{\mu_{B}\log_{2}\Bigg(\frac{1}{\mu_{B}}\Bigg)}\Bigg|\nonumber\\
&\quad=\Bigg|\sum_{A\in\mathcal{A}}\sum_{B\in\mathcal{A}}{\mu_{AB}\log_{2}\Bigg(\frac{\mu_{B}}{\mu_{A}}\Bigg)}\Bigg|=\Bigg|\sum_{(A,B)\in E}{\mu_{AB}\log_{2}\Bigg(\frac{\mu_{B}}{\mu_{A}}\Bigg)}\Bigg|\nonumber\\
&\quad\leq\log_{2}\Bigg(\sum_{(A,B)\in E}{\Bigg(\frac{\mu_{AB}\mu_{B}}{\mu_{A}}\Bigg)}\Bigg)\leq\log_{2}\Bigg(\sum_{(A,B)\in E}\frac{\mu_{AB}}{\mu_{A}}\Bigg)\nonumber\\
&\qquad=\log_{2}\Bigg(\sum_{B\in\partial(A)}\mathbf{P}\{Y\in B\mid X\in A\}\Bigg)\leq \log_{2}(\Delta+1).\nonumber
\end{align}
The quantity ``1'' within $\log_{2}(\Delta+1)$ comes from the fact that we have always $A\in\partial(A)$. If $T$ is the random number of bits needed to generate a discrete random variable $Y$ that outputs a vertex $A$ of $G$ with probability $\mathbf{P}\{Y\in A\}$, then
\begin{align}
\mathbf{E}(T)\geq H(Y\mid\mathcal{A}) \geq H(X\mid\mathcal{A})-\log_{2}(\Delta+1).\label{onara_eq_lower}
\end{align}
Because the inequality from (\ref{onara_eq_lower}) is valid for all choice of $\mathcal{A}$, then
\begin{align*}
\mathbf{E}(T)&\geq \sup_{\mathcal{A}}\big\{H(X\mid\mathcal{A})-\log_{2}(\Delta+1)\big\},
\end{align*}
which ends the proof.
\end{proof}

\begin{remark}
The bound from lemma \ref{lowerbndlem} is equal to Shannon's bound \cite{Sha_1948} when the distribution $X$ is discrete with a finite support, in that case, $\Delta=0$ by choosing $\epsilon$ sufficiently small.
\end{remark}

We recall a result from Csisz\`{a}r \cite{Csi_1962} about the hypercubic partition entropy of an absolutely continuous random vector $X$. Of particular interest to us is the hypercubic partition, denoted by $\mathcal{A}_{h}^{\star}$ for some $h>0$, and for which a cell $A\in\mathcal{A}_{h}^{\star}$ has the form
\begin{displaymath}
\prod_{j=1}^{d}{\big[i_{j}h,(i_{j}+1)h\big)\quad\text{for some }(i_{1},\ldots,i_{d})}\in\mathbb{Z}^{d}.
\end{displaymath}
Integer vectors $(i_{1},\ldots,i_{d})\in\mathbb{Z}^{d}$ are used if necessary to index cells in $\mathcal{A}_{h}^{\star}$. We recall that $\lfloor X\rfloor = (\lfloor X_1\rfloor,\ldots,{}\lfloor X_d\rfloor)$ has a finite entropy, a condition we refer to as R\'{e}nyi's condition. If $F$ is absolutely continuous with density $f$, then
\begin{displaymath}
H(f)=-\int{\log_{2}(f){d}F}
\end{displaymath}
is well-defined, that is either finite or $-\infty$.
\begin{theorem}[Csisz\`{a}r \cite{Csi_1962}]\label{csiszar_lemma}
Let $X\in\mathbb{R}^{d}$ be a continuous random variable with distribution $F$. If $F$ is absolutely continuous with density $f$, and the R\'{e}nyi's condition is satisfied, then
\begin{equation}
H(X\mid\mathcal{A}_{h}^{\star})-d\log_{2}\bigg(\frac{1}{h}\bigg)\to H(f)\quad\text{as}\quad h\to 0^{+}\label{hippo_farts}
\end{equation}
If $H(f)$ is bounded, then the right side of (\ref{hippo_farts}) is finite. If $H(f)=-\infty$, then the right side of (\ref{hippo_farts}) is $-\infty$.
If $F$ is singular, then
\begin{displaymath}
H(X\mid\mathcal{A}_{h}^{\star})-d\log_{2}\bigg(\frac{1}{h}\bigg)\rightarrow -\infty\quad\text{as}\quad h\to 0^{+}.
\end{displaymath}
\end{theorem}

R\'{e}nyi \cite{Ren_1959}, Csisz\`{a}r \cite{Csi_1961}, and Linder and Zeger \cite{LinZeg_1994} contain further information about the asymptotic theory of entropy arising from general partitions.

\begin{lemma}\label{absctsrenyithm}
Under R\'{e}nyi's condition, for general partition $\mathcal{A}$ and a random variable $X\in\mathbb{R}^{d}$ with density $f$, we have
\begin{displaymath}
H(X\mid\mathcal{A})\geq H(f)+\sum_{A\in\mathcal{A}}{\mathbf{P}\{X\in A\}\log_{2}\bigg(\frac{1}{\lambda(A)}\bigg)},
\end{displaymath}
where $\lambda$ denotes the Lebesgue measure. In particular, we have
\begin{displaymath}
H(X\mid\mathcal{A}_{h}^{*})\geq H(f)+d\log_{2}\bigg(\frac{1}{h}\bigg).
\end{displaymath}
\end{lemma}
\begin{proof}[Proof of lemma \ref{absctsrenyithm}]
Given $A\in\mathcal{A}$, let $Z$ be uniformly distributed on $A$ so that, for any $B\subset A$, we have $\mathbf{P}\{z\in B\}=\lambda(B)/\lambda(A)$. Then we have as well that
\begin{displaymath}
\mathbf{E}(f(Z))=\int_{z\in A}{f(z){d}\mathbf{P}(z)}=\frac{1}{\lambda(A)}\int_{x\in A}{f(x){d}x}=\frac{\mathbf{P}\{X\in A\}}{\lambda(A)}.
\end{displaymath}
For notational convenience, we write $Y=f(Z)$, and then
\begin{displaymath}
\frac{\mathbf{P}\{X\in A\}}{\lambda(A)}\log_{2}\bigg(\frac{\lambda(A)}{\mathbf{P}\{X\in A\}}\bigg)=\mathbf{E}(Y)\log_{2}\bigg(\frac{1}{\mathbf{E}(Y)}\bigg).
\end{displaymath}
By Jensen's inequality and the concavity of~\mbox{$-x\log_{2}(x)$},
\begin{displaymath}
\mathbf{E}(Y)\log_{2}\bigg(\frac{1}{\mathbf{E}(Y)}\bigg)\geq\mathbf{E}\bigg(Y\log_{2}\bigg(\frac{1}{Y}\bigg)\bigg)=\frac{1}{\lambda(A)}\int_{A}{f\log_{2}\bigg(\frac{1}{f}\bigg)},
\end{displaymath}
and the inequality follows by summing over $A\in\mathcal{A}$.
\end{proof}

\begin{theorem}\label{lowerbndthm}
Let $X\in\mathbb{R}^{d}$ be a random variable with density $f$ and let $Y$ be an output of $(\epsilon,p)$-sampling algorithm for $X$ that requires to process $T$ random bits in order to produce $Y$. If R\'{e}nyi's condition is true, then
\begin{displaymath}
\mathbf{E}(T)\geq H(f)+d\log_{2}\bigg(\frac{1}{\epsilon}\bigg)-\log_{2}V_{d,p}\quad\text{with}\quad V_{d,p}=\frac{2^{d}\Gamma\big(\frac{1}{p}+1\big)}{\Gamma\big(\frac{d}{p}+1\big)},
\end{displaymath}
and the latter quantity is the volume of the unit ball in $\mathbb{R}^{d}$.
\end{theorem}
\begin{proof}[Proof of theorem \ref{lowerbndthm}]
Let $\mathcal{A}_{h}^{\star}$ be a cubic partition for some $h>0$. Using lemma \ref{lowerbndlem}, we have
\begin{displaymath}
\mathbf{E}(T)\geq \sup_{h>0}\Big(H(X\mid\mathcal{A}_{h}^{\star})-\log_{2}\big(\Delta_{h}+1\big)\Big)
\end{displaymath}
where $\Delta_h$ is the maximal degree of the sampling graph with vertices in \mbox{$\mathcal{A}_{h}^{\star}\times\mathcal{A}_{h}^{\star}$}. Also \mbox{$\inf\{\|x-y\|_{p}<\epsilon\hspace{1pt}\colon\hspace{1pt}x\in A,y\in B\}$} if and only if $(A,B)\in\mathcal{A}_{h}^{\star}\times\mathcal{A}_{h}^{\star}$ defines an edge of the sampling graph. We set $h=\epsilon\slash n$ and use
\begin{displaymath}
\mathbf{E}(T)\geq \limsup_{n\to\infty}\big(H(X\mid\mathcal{A}_{\epsilon\slash n}^{\star})-\log_{2}\big(\Delta_{\epsilon\slash n}+1\big)\big).
\end{displaymath}
If $B_{r}$ denotes the $\ell_{p}$-ball of radius $r$ centered at $0$, then by elementary considerations on sphere packing,
\begin{displaymath}
\frac{\lambda(B_{\epsilon})}{h^{d}}\leq \Delta_{h}\leq \frac{\lambda\big(B_{\epsilon'}\big)}{h^{d}}\quad\text{where }\epsilon'=\epsilon+2hd^{1\slash p},
\end{displaymath}
so that as $n\to\infty$,
\begin{displaymath}
\Delta_{\epsilon\slash n}\sim \bigg(\frac{n}{\epsilon}\bigg)^{d}\lambda(B_{\epsilon})=V_{d,p}n^{d}.
\end{displaymath}
We have also that $
H(X\mid\mathcal{A}_{\epsilon\slash n}^{\star})\geq H(f)+d\log_{2}\big(\frac{n}{\epsilon}\big)$,
and we conclude the proof because
\begin{align*}
&H(X\mid\mathcal{A}_{\epsilon\slash n}^{\star})-\log_{2}\big(\Delta_{\epsilon\slash n}+1\big)\\
&\quad\geq H(f)+d\log_{2}\bigg(\frac{1}{\epsilon}\bigg)+\log_{2}\bigg(\frac{n^{d}}{1+V_{d,p}n^{d}(1+o(1))}\bigg)\\
&\quad\stackrel{n\to\infty}{\rightarrow}H(f)+d\log_{2}\bigg(\frac{1}{\epsilon}\bigg)-\log_{2}V_{d,p}.
\end{align*}
\end{proof}

We end this section with a few observations and examples. One of the observations is that the concept of Wasserstein metric may not be always needed in order to obtain arbitrarily accurate random outcomes. As a first example of a very-easy-to-generate continuous distributions, we can take the Cantor singular distribution defined over the unit length interval $[0,1]$. We recall that a random variable distributed according to the Cantor distribution has only $0$ or $2$ in its triadic expansion. In fact, a $(\epsilon,\infty)$-sampling algorithm for the Cantor distribution consists to generate a sequence of $\ell=\lceil-\log_{3}(\epsilon)\rceil$ random bits $X_1, \ldots, X_{\ell}$, and to output $X=(2X_{1})3^{-1}+\ldots + (2X_{\ell})3^{-\ell}$. Then $X$ is distributed accordingly to the Cantor distribution with precision $\epsilon$. A second easy-to-generate continuous distribution is the truncated exponential upon which a standard (non-truncated) exponential can be generated. Indeed the truncated exponential can be written as the convolution of Bernoulli random variables that is a the sum of independent random variables. The last two examples bring us to mention a third observation which is that the Bernoulli random variables $X_i$'s are not identically distributed. We defer actually to section \ref{sect_convolution} an interesting theorem due to Kakutani which characterizes distributions obtained by convolution, and which turns out to be useful to generate random variables by convolution. We also point out a fourth observation that when a (continuous) distribution is expressed as a linear combination of independent, not necessarily identical, discrete random variables, then batch generation still obviously works; we can therefore extract randomness as we are generating sequentially the $X_{i}$'s for $1\leq i\leq \ell$ to reach asymptotically the optimal complexity for the number of expected bits required.

\subsection{Wasserstein metric and probability integral transform}\label{sect_wasserstein_inversion}

We establish an interesting connection between the Wasserstein $L_{\infty}$-metric between two distributions and the probability integral transform or inverse of the latter two distributions. Our theorem applies to one-dimensional distribution and more research is needed to extend it to multi-dimensional distribution. Our theorem here further supports the choice of the Wasserstein $L_{\infty}$-metric as the natural metric in the realm of non-uniform random variate generation to measure accuracy. Given two random variables $X$ and $Y$ with distributions $F$ and $G$ respectively, let $\mathcal{M}$ be the class of joint distributions defined over $\mathbb{R}^{2}$ for which their marginal distributions coincide with $F$ and $G$. We recall the $\ell_{\infty}$-Wasserstein distance between $F$ and $G$ given by
\begin{align*}
W_{\infty}(F,G)&= \inf\{\esssup\|X-Y\|_{\infty}\colon(F,G)\in\mathcal{M}\}\\
&=\inf\{\esssup|X-Y|\colon(F,G)\in\mathcal{M}\}.
\end{align*}
From now on, we write simply $W(F,G)$.

\begin{theorem}\label{wasser_inversion_thm}
With the notation of before, we have
\begin{displaymath}
W(F,G)=\sup_{u\in(0,1)}|F^{-1}(u)-G^{-1}(u)|.
\end{displaymath}
\end{theorem}
\begin{proof}
Let $U$ be a continuous uniform random variable on the real interval $(0,1)$. By the probability integral transform, we have $F^{-1}(U) \stackrel{\mathcal{D}}{=} X$ and $G^{-1}(U) \stackrel{\mathcal{D}}{=}Y$. Therefore, we have
\begin{equation}
W(F,G)\leq \esssup|F^{-1}(U)-G^{-1}(U)|\leq \sup_{u\in[0,1]}|F^{-1}(u)-G^{-1}(u)|.\label{namatoragaluptere}
\end{equation}
If $F$ or $G$ are constant on some sub-intervals, then, on any set with non-zero measure, (\ref{namatoragaluptere}) holds for any value $F^{-1}(u)$ such that
\begin{displaymath}
\inf\{x\colon F(x)\geq u\}\leq F^{-1}(u) \leq \sup\{x\colon F(x)\geq u\}.
\end{displaymath}
To show also that $\sup_{u\in[0,1]}|F^{-1}(u)-G^{-1}(u)|\leq W(F,G)$ on any set with non-zero measure, we proceed by contradiction. We recall that
\begin{align*}
F(x)&=\mathbf{P}\{X\leq x\},\hspace{12pt} F^{-1}(u)=\inf\{x\colon F(x)\geq u\},\\
G(y)&=\mathbf{P}\{Y\leq x\},\hspace{12pt} G^{-1}(u)=\inf\{y\colon G(y)\geq u\}.
\end{align*}
Given $\delta>0$ such that $W(F,G)>\delta$, then the pair $(X,Y)$ satisfies
\begin{align}
\mathbf{P}\{|X-Y|\leq \delta\}=1.\label{blatumividarol}
\end{align}
We extend the joint distribution of $(X,Y)$ over sets of measure zero in such a way that $|X-Y|\leq\delta$; we derandomize (\ref{blatumividarol}) in other words. Then we can obtain a contradiction if we suppose that $F^{-1}(u)=G^{-1}(u)-\theta\hspace{4pt}$ and $\theta>\delta$. For convenience, let us write $x_{u}=F^{-1}(u)$ and $y_{u}=G^{-1}(u)$. Graphically we have the situation represented on figure \ref{sedatuminolex}.

\begin{figure}
\begin{centering}
\includegraphics{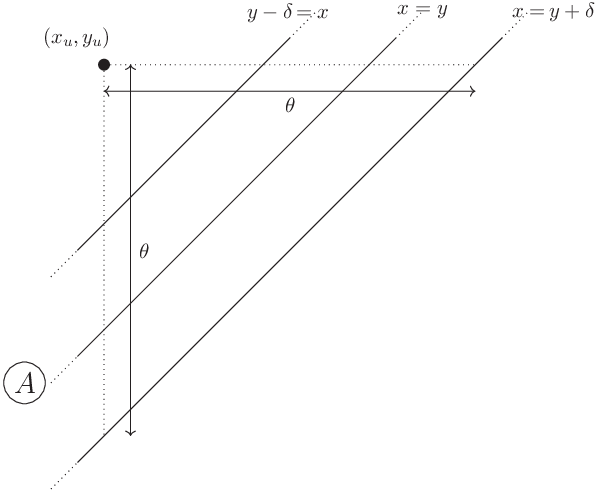}
\caption{Wasserstein and the probability integral transform}\label{sedatuminolex}
\end{centering}
\end{figure}

Let $A=\{(x,y)\in\mathbb{R}^{2}\colon |x-y|\leq \delta\}$ as represented on figure \ref{sedatuminolex} so that the plane is split into regions defined by the following probabilistic relations:
\begin{align*}
\mathbf{P}\{(X,Y)\in A\}&=1,\\
\mathbf{P}\{(X,Y)\in A\cap \big((-\infty,x_u)\times\mathbb{R}\big)\}&=u,\\
\mathbf{P}\{(X,Y)\in A\cap \big(\mathbb{R}\times (y_u,\infty)\big)\}&=1-u.
\end{align*}
Therefore we have
\begin{align*}
\mathbf{P}\{(X,Y)\in A\cap \big((x_u,\infty)\cap (-\infty,y_u)\big)\}&=0,\\
\mathbf{P}\{Y\leq y_u\}=\mathbf{P}\{Y\leq x_u+\delta\}&\Longrightarrow G(y_u)=G(x_u+\delta),
\end{align*}
and since, by definition $G^{-1}(u)=\inf\{y\colon G(y)\geq u\}$, we have that $y_u\neq G^{-1}(u)$ which is a contradiction.

We end the proof by extending to all pair $(X,Y)\in\mathcal{M}$ by the use of an approximation based argument. For that, let $\epsilon>0$ and $U$ uniformly distributed over $(0,1)$ independent of $X$ and $Y$. Define the random variables $X_{\epsilon}=X+\epsilon U$ and $Y_{\epsilon}=Y+\epsilon U$ with respective distributions $F_{\epsilon}$ and $G_{\epsilon}$. We observe that
\begin{align}
W(F_{\epsilon},G_{\epsilon})=\esssup|F_{\epsilon}^{-1}(U)-G_{\epsilon}^{-1}(U)|.\label{trimunotriti}
\end{align}
If $X$ and $Y$ are paired in such a way that $\esssup|X-Y|=W(F,G)$, and since $|X-X_{\epsilon}|\leq \epsilon$ and $|Y-Y_{\epsilon}|\leq \epsilon$, then by the triangle inequality,
\begin{displaymath}
\big|\esssup|X_{\epsilon}-Y_{\epsilon}|-\esssup|X-Y|\big|\leq 2\epsilon.
\end{displaymath}
The proof is completed because $|F_{\epsilon}^{-1}(U)-F^{-1}(U)|\leq\epsilon$ and $|G_{\epsilon}^{-1}(U)-G^{-1}(U)|\leq\epsilon$ which combines with (\ref{trimunotriti}) yields to
\begin{displaymath}
\esssup|F^{-1}(U)-G^{-1}(U)|\leq W(F_{\epsilon},G_{\epsilon})\leq\esssup|F^{-1}(U)-G^{-1}(U)|+4\epsilon.
\end{displaymath}
\end{proof}

\subsection{Interlude about the differential entropy}\label{sect_diff_ent_usage}

Due to the importance that the differential entropy plays in the remainder of our work, we deemed necessary to discuss when the use of the differential entropy is justified in the context of the generation of continuous random variables. This section illustrates actually a somewhat counter intuitive fact. Namely the fact that there exists absolutely continuous distributions with bounded differential entropy, but for which infinitely many possible discretizations yield to discrete probability distributions with unbounded entropy. We show how to construct an instance of such probability density function. The concept of differential entropy of a density is not meaningful in the context of random number generation whenever the entropy corresponding to one of its non-trivial discretization diverges.

Let $k>0$ be an integer and let $\mathbf{p}=(p_1,p_2,\ldots)$ be an infinite length probability vector such that
\begin{displaymath}
\sum_{k=1}^{\infty}{p_k}=1\quad\text{and}\quad\sum_{k=1}^{\infty}{p_k\log_{2}\bigg(\frac{1}{p_k}\bigg)}=\infty.
\end{displaymath}
For instance, a choice for $\mathbf{p}$ can be any probability vector from the family of Zeta-Dirichlet distributions parameterized by a real number $u>0$, and for which
\begin{displaymath}
C_{u}=\sum_{k=3}^{\infty}{\frac{1}{k(\log{k})^{1+u}}}\quad\text{and}\quad
p_{k}=\frac{1}{C_{u}}\frac{1}{k(\log{k})^{1+u}}\quad\text{for $k\geq 3$.}
\end{displaymath}
As shown in Hardy and Riesz \cite{HarRis_BookNewEd2005}, the entropy is unbounded for all $0<u\leq 1$ and bounded for $u>1$. For $u\leq 0$, the sum $C_{u}$ diverges.

For the construction, let $\mathbf{a}=(a_{k})_{k\geq 1}$ be an increasing sequence of real numbers such that $a_{k}+p_{k}\leq a_{k+1}$. Let $\mathcal{A}=\{A_{k}\}_{k\geq 1}$ be a family of disjoints subsets of $\mathbb{R^{+}}$ such that $A_{k}=[a_{k},a_{k}+p_{k})$. We observe that $\mathcal{A}$ may or may not cover $\mathbb{R^{+}}$. We now define the density of our random variable of interest, say $X$, as well as its law as follow. We denote its density by $f$ and define $f:\mathbb{R}^{+}\to I$ for some $I\subset[0,\infty)$ through $\mathcal{A}$ as $f(x)=\sum_{k=1}^{\infty}{\mathds{1}\{x\in A_{k}\}}$. Clearly we have for all $x\in\mathbb{R}$ that $f(x)\in\{0,1\}$, and so $I=\{0,1\}$. The cumulative distribution function $F$ is given by
\begin{align*}
F(x)&=\sum_{k=1}^{\infty}{(x-a_k+q_{k-1})\mathds{1}\{x\in A_k\}}+\sum_{k=0}^{\infty}{q_k\mathds{1}\{x\in B_k\}}\\
q_{0}&=0,\quad B_{0}=(-\infty,a_1),\quad\text{and}\\
q_{k}&=\sum_{j=1}^{k}{p_{j}},\quad B_{k}=[a_k+p_k,a_{k+1})\quad\text{for $k\geq 1$,}
\end{align*}
For $k\geq 1$, if $B_{k}=\emptyset$, then there is no gap between $A_{k}$ and $A_{k+1}$ that is $a_k+p_k=a_{k+1}$ and so $\mathcal{A}$ covers $\mathbb{R}^{+}$ into disjoint subsets. To get an unbounded entropy from a discretization, we choose $\mathcal{A}$ such that there are infinitely many consecutive intervals $A_{k}$ and $A_{k+1}$ with gaps, that is with $B_k\neq \emptyset$. Figure \ref{fig_dist_entrop_infini} represents $F$ around some interval $B_{k}$ for some $k>0$.

\begin{figure}
\begin{centering}
\includegraphics{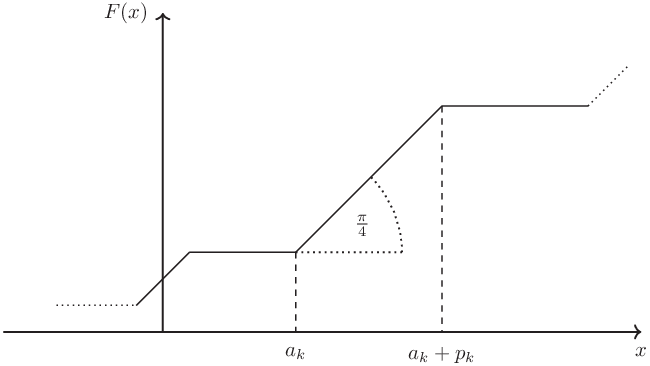}
\caption{\label{fig_dist_entrop_infini}Distribution $F$ shown around an interval $B_{k}$ for some $k\geq 1$.}
\end{centering}
\end{figure}

The differential entropy is $0$ since
\begin{displaymath}
H(f)=\int_{-\infty}^{+\infty}{f(x)\log_{2}\bigg(\frac{1}{f(x)}\bigg)dx}=\sum_{k=1}^{\infty}{\int_{x\in A_{k}}{f(x)\log_{2}\bigg(\frac{1}{f(x)}\bigg)dx}}=0.
\end{displaymath}

Now let $\mathcal{R}$ be a partition covering $\mathbb{R^{+}}$ into disjoints intervals. More precisely, choose arbitrarily a real $\epsilon>0$ and let $\mathcal{R}=\{R_{i}\}_{i \geq 0}$ be such that $R_{i}=[i\epsilon,(i+1)\epsilon)$. For $\mathcal{A}$ with infinitely many consecutive intervals $A_{k}$ and $A_{k+1}$ having a non-zero gap, then there exists a sequence of indices $k_{j}$ with $j\geq 0$ such that $a_{k_{j}}+p_{k_{j}}<a_{k_{j}+1}$. For all $\epsilon>0$, there is $K>0$ such that for all $k\geq K$, $p_{k}<\epsilon$ since the $p_{k}$'s are the terms of a convergent series. Because $\mathcal{R}$ covers $\mathbb{R}^{+}$ into disjoint sets, there exists an index $i_{K}$ such that for all $i\geq i_{K}$ we have $[a_{i},a_{i}+p_{i})\subset R_{i}=[i\epsilon,(i+1)\epsilon)$, which is always possible by the choice of $\mathcal{A}$. Therefore the entropy of the discretization of $X$ under $\mathcal{R}$ is given by
\begin{align*}
&H(X\hspace{1pt}|\hspace{1pt}\mathcal{R})=\sum_{i=0}^{\infty}{\mathbf{P}\{X\in R_{i}\}\log_{2}\bigg(\frac{1}{\mathbf{P}\{X\in R_{i}\}}\bigg)}\\
&\quad>\sum_{i=i_{K}}^{\infty}{\mathbf{P}\{X\in R_{i}\}\log_{2}\bigg(\frac{1}{\mathbf{P}\{X\in R_{i}\}}\bigg)}=\sum_{i=i_{K}}^{\infty}{p_i\log_{2}\bigg(\frac{1}{p_i}\bigg)}=\infty.
\end{align*}

\begin{remark}
As explained in previous sections, we should therefore expect an unbounded number of random bits to generate an instance $X$ with at least $\lceil\log_{2}(1/\epsilon)\rceil$ bits of accuracy. Due to the Chebeshev's inequality apply to the tail of the probability of halting, in practice it may however terminates with a significant high probability depending of the target distribution.
\end{remark}

We end this section by mentioning a sufficient condition for the convergence of the differential entropy.
\begin{lemma}\label{lem:conv_diff_ent}
Let $X\in\mathbb{R}$ be an absolutely continuous random variable with density $f$.
\begin{displaymath}
\text{If }\phantom{1}\mathbf{E}\big(\log_{2}(1+|X|)\big)<\infty \text{ then }\phantom{1}\int_{\mathbb{R}}{f(x)\log_{2}\bigg(\frac{1}{f(x)}\bigg)dx}<\infty.
\end{displaymath}
\end{lemma}
\begin{proof}
Let $A=\{x\in\mathbb{R}\colon f(x)\leq 1\}$ and $g$ be the density for a centered-around-zero Cauchy random variable. We recall that $g$ is given by
\begin{displaymath}
g(x)=\frac{1}{\pi}\frac{1}{1+x^{2}}\text{ for $x\in\mathbb{R}$}.
\end{displaymath}
We have that
\begin{displaymath}
\big(1+|x|\big)^{2}\geq 1+|x|^{2}=1+x^{2}\quad\text{and}\quad t\log_{2}\bigg(\frac{1}{t}\bigg)\leq \frac{1}{e\log(2)}\quad\text{for all $t>0$}.
\end{displaymath}
If $\mathbf{E}\big(\log_{2}(1+|X|)\big)<\infty$, then
\begin{align*}
&\int_{\mathbb{R}}{f(x)\log_{2}\bigg(\frac{1}{f(x)}\bigg)dx}< \int_{A}{f(x)\log_{2}\bigg(\frac{1}{f(x)}\bigg)dx}\\
&\quad=\int_{A}{f(x)\log_{2}\bigg(\frac{1}{g(x)}\frac{g(x)}{f(x)}\bigg)dx}\\
&\quad=\int_{A}{f(x)\log_{2}\bigg(\frac{1}{g(x)}\bigg)dx}+\int_{A}{g(x)\frac{f(x)}{g(x)}\log_{2}\bigg(\frac{g(x)}{f(x)}\bigg)dx}\\
&\quad< \int_{\mathbb{R}}{\log_{2}(\pi) f(x)dx}+2\mathbf{E}\big(\log_{2}(1+|X|)\big)+\frac{1}{e\log(2)}\int_{\mathbb{R}}{g(x)dx}<\infty.
\end{align*}
\end{proof}

\subsection{Discretization}\label{sect_discretization}

We explain in this section the complexities that we should expect when we apply a DDG-tree based algorithms as in section \ref{sect_HH} to generate a discretized continuous random variable. A discretization is a partition as explained in section \ref{sect_cts_gen}. We recall that the expected complexity of a generic DDG-tree algorithm $\mathtt{A}$ to generate a discrete random variable, say $Y$, is bounded above by $H(Y)+C$ for some $C>0$. The constant $C$ is related to the expected conditional entropy of the depths at which $\mathtt{A}$ outputs an instance of $Y$.

Consider a continuous $d$-dimensional random variable $X$ with support $I\subset\mathbb{R}^{d}$. Let $\epsilon>0$, and $\mathcal{A}_{\epsilon}$ be a partition of $I$, as in section \ref{sect_cts_gen}, that defines an $(\epsilon,p)$-sampling graph. In addition, we denote a center of $A\in\mathcal{A}_{\epsilon}$ by $x_A$. The center might not be unique especially if $X$ is singular. By definition we have that $\sup\{\|x_A-y\|_{p}\colon y\in A\}\leq \epsilon$. A sampling algorithm for the discrete distribution $\{\mathbf{P}(X\in A)\}_{A\in\mathcal{A}}$ that returns $A\in\mathcal{A}_{\epsilon}$ can be used to generate a random variable $Y$ that approximate $X$ to within $\epsilon$. Indeed, once the sampling algorithm returns $A$, we only need to set $Y=x_{A}$ and return $Y$. Thus there is a coupling $(X,Y)$ with \mbox{$\|X-Y\|_{p}\leq \epsilon$}.

Now for the rest of this section, we assume that the distribution of $X$ is absolutely continuous and we analyse two cases: $p=\infty$ or $1\leq p <\infty$. For $p=\infty$, a good choice is the hypercubic partition with sides $2\epsilon$, that is $\mathcal{A}_{2\epsilon}^{\star}$. If the distribution of $X$ has density $f$ and satisfies R\'{e}nyi's condition with \mbox{$H(f)>-\infty$}, then we have
\begin{align}
\mathbf{E}(T)\leq H(Y)+C&\leq H(f) + d\log_{2}\bigg(\frac{1}{2\epsilon}\bigg)+C+o(1)\nonumber\\
&\quad=H(f) + d\log_{2}\bigg(\frac{1}{\epsilon}\bigg)+C-d+o(1).\label{up_bnd_discret_machin}
\end{align}
We compare the lower bound from theorem \ref{lowerbndthm} with (\ref{up_bnd_discret_machin}) to observe a difference of $C+o(1)$. We also observe that if $d=1$, the a simple partition into intervals of length $2\epsilon$ can be used for all values of $p$ to obtain (\ref{up_bnd_discret_machin}).

For general $p\in[1,\infty)$, if we fix the $\ell_{p}$-balls radius to $2\epsilon{}d^{-\frac{1}{p}}$, then we observe unfortunately a linear growth in $d$ for the expected complexity which we do not have for $p=\infty$ as just explained. As before we assume that $f$ is the density of $X$ and that R\'{e}nyi's condition with $H(f)>-\infty$. Then we have
\begin{align}
\mathbf{E}(T)\leq H(Y)+C&\leq H\Big(X\mid \mathcal{A}_{2\epsilon{}d^{-\frac{1}{p}}}^{\star}\Big)+2\nonumber\\
&\quad\leq d\log_{2}\bigg(\frac{1}{\epsilon}\bigg)+H(f)+C-d+\frac{d}{p}\log_{2}(d)+o(1).\label{caca_collant}
\end{align}
Using $\Gamma\big(1+u\big)\geq \big(u\slash e\big)^{u}\sqrt{2\pi u}$, $u>0$, the difference between (\ref{caca_collant}) and the lower bound from theorem \ref{lowerbndthm} is
\begin{align*}
&=2+\frac{d}{p}\log_{2}(d)+d\log_{2}\Gamma\bigg(\frac{1}{p}+1\bigg)-\log_{2}\Gamma\bigg(\frac{d}{p}+1\bigg)+o(1)\\
&\leq 2+d\log_{2}\bigg(\Gamma\bigg(\frac{1}{p}+1\bigg)(ep)^{\frac{1}{p}}\bigg)-\frac{1}{2}\log_{2}\bigg(2\pi \frac{d}{p}\bigg)+o(1),
\end{align*}
which unfortunately increases linearly with $d$. To avoid this growing differential, it seems necessary to consider partitions that better approximate $\ell_{p}$-balls, a topic of further research.

\subsection{Inversion}\label{sect_inversion}

We analyze in this section the expected complexity for the inversion method. For a one-dimensional random variable $X$ with distribution function $F$, we recall that $X\stackrel{\mathcal{D}}{=}F^{-1}(U)$ where $F^{-1}$ denotes the inverse of $F$, and $U$ is uniformly distributed on $[0,1)$. Because $F$ is a right-continuous and increasing, we can generate a sequence of approximations for $X$ for which consecutive approximations get closer and closer to $X$ as we are invoking more and more \texttt{RandomBit}.

If $(U_{i})_{i>0}$ is a sequence of unbiased, independently and identically distributed random bits, then let
\begin{align*}
U&=0.U_1U_2\cdots=\sum_{j=1}^{\infty}{\frac{U_j}{2^{j}}}\quad\text{for}\quad U\in[0,1),\\
U_{(t)}&=0.U_1\cdots U_{t}\quad\text{and}\quad U_{(t)}^{+}=0.U_1\cdots U_{t}+\frac{1}{2^{t}}=0.U_{1}\cdots U_{t}1^{\infty}.
\end{align*}
We have $U_{(t)}\leq U\leq U_{(t)}^{+}$, $U_{(0)}=0$ and $U_{(1)}=1$. As a visual illustration, we have figure \ref{fig_inversion}.
\begin{figure}
\begin{centering}
\includegraphics{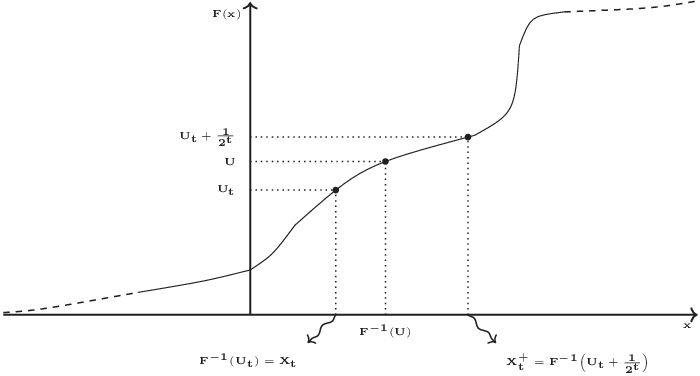}
\caption{Inversion method illustrated}\label{fig_inversion}
\end{centering}
\end{figure}
For $\epsilon>0$, if we define
\begin{displaymath}
Y=\frac{F^{-1}\big(U_{(t)}^{+}\big)+F^{-1}\big(U_{(t)}\big)}{2},
\end{displaymath}
then $X$ and $Y$ are coupled in such a way that $|X-Y|\leq 2\epsilon$. We need to analyze the expected value of $T_{\epsilon}$ defined by
\begin{displaymath}
T_{\epsilon}=\min\{t\geq 0\phantom{1}:\phantom{1}F^{-1}\big(U_{(t)}^{+}\big)-F^{-1}\big(U_{(t)}\big)\leq 2\epsilon\}.
\end{displaymath}

The inversion method mimics in spirit the method of Han and Hoshi, and indeed, this observation leads to a simple bound. Let $\mathcal{A}_{2\epsilon}^{\star}$ be a partition of $\mathbb{R}$ into disjoint intervals of equal length $2\epsilon$. Denote the probabilities of these intervals by $\mathbf{P}\{X\in A\}$ for $A\in\mathcal{A}_{2\epsilon}^{*}$ and its corresponding probability vector $\mathbf{p}$. Assume that we select randomly an interval according to $\mathbf{p}$ using the method of Han and Hoshi. It is easy to see that the number of bits needed to halt in the inversion method is smaller. Therefore, for the inversion method, we have also that $\mathbf{E}(T)\leq H(X\mid \mathcal{A}_{2\epsilon}^{\star})+3$. From the aforementioned upper bound and lemma \ref{csiszar_lemma}, we conclude theorem \ref{inversion_borne_sup_thm} that follows immediately.

\begin{theorem}\label{inversion_borne_sup_thm}
If $X$ has a density $f$ satisfying R\'{e}nyi's condition and if $H(f)>-\infty$, then, as $\epsilon\to 0^{+}$,
\begin{align}
\mathbf{E}(T)\leq \log_{2}\bigg(\frac{1}{\epsilon}\bigg)+H(f)+2+o(1)\label{caca_juteux}
\end{align}
\end{theorem}
In addition, one can tighten the analysis under additional conditions on $f$ such unimodality, monotonicity, or for specific forms. For that we the following theorem.
\begin{theorem}\label{thm_inv_mono_bornee}
Assume that $X$ has a bounded nonincreasing density $f$ on $[0,\infty)$. For the inversion method described above, if $\epsilon\to 0^{+}$, then
\begin{displaymath}
\mathbf{E}(T)\leq \log_{2}\bigg(\frac{1}{\epsilon}\bigg)+H(f)+o(1).
\end{displaymath}
\end{theorem}
\begin{proof}[Proof of theorem \ref{thm_inv_mono_bornee}]
Define $X_{t}=F^{-1}\big(U_{(t)}\big)$, $X_{t}^{+}=F^{-1}\big(U_{(t)}^{+}\big)$ as in figure \ref{fig_inversion}. Then
\begin{align*}
\mathbf{E}(T)&=\sum_{t=0}^{\infty}{\mathbf{P}\big\{X_{t}^{+}-X_{t}>2\epsilon\big\}}\leq \sum_{t=0}^{\infty}{\mathbf{P}\Big\{f(X_{t}^{+})<\frac{1}{2^{t}2\epsilon}\Big\}}\\
&\quad\leq\sum_{t=0}^{\infty}{\mathbf{P}\Big\{f(X)<\frac{1}{2^{t}2\epsilon}\Big\}}+\sum_{t=0}^{\infty}{\mathbf{P}\Big\{f(X_{t}^{+})<\frac{1}{2^{t}2\epsilon}<f(X)\Big\}}=\textrm{I}+\textrm{II}.
\end{align*}
Now we observe that
\begin{align*}
\textrm{I}&\leq\mathbf{E}\bigg(1+\log_{2}\frac{1}{2\epsilon f(X)}\bigg)=\log_{2}\bigg(\frac{1}{\epsilon}\bigg)+H(f),
\end{align*}
even if $H(f)=\infty$. The theorem follows if we can show that $\textrm{II}=o(1)$. To this end, note that
\begin{displaymath}
\textrm{II}\leq \sum_{t=0}^{\infty}{\mathbf{P}\Big\{f(X_{t}^{+})<\frac{1}{2^{t}2\epsilon}\leq f(X_{t})\Big\}}.
\end{displaymath}
For a fixed value of $t$, we see that $f(X_{t}^{+})<\frac{1}{2^{t}2\epsilon}\leq f(X_{t})$ only if $X$ falls in the interval that ``captures'' the value $\frac{1}{2^{t}2\epsilon}$, if such an interval exists. The probability of each interval is precisely $1\slash 2^{t}$. If $\frac{1}{2^{t}2\epsilon}>f(0)$, then no such interval exists from which the proof is completed since
\begin{align*}
\textrm{II}&\leq \sum_{t=0}^{\infty}{\frac{1}{2^{t}}\mathds{1}\bigg\{t\geq \log_{2}\Big(\frac{1}{2\epsilon f(0)}\Big)\bigg\}}\leq 4\epsilon f(0) = o(1).
\end{align*}
\end{proof}

In the next two sections, we evaluate the expected complexities of the inversion method for specific distributions, namely the exponential distribution with mean one and the standard normal distribution.

\subsubsection{The exponential law}
For the exponential density, the inversion method yields
\begin{displaymath}
\mathbf{E}(T)\leq \log_{2}\bigg(\frac{1}{\epsilon}\bigg)+H(f)+o(1)\leq \log_{2}\bigg(\frac{1}{\epsilon}\bigg)+\log_{2}(e)+o(1),
\end{displaymath}
where $\log_{2}(e)=1.443\ldots$ Flajolet and Saheb \cite{FlaSah_1986} proposed a method for the exponential law that has
\begin{displaymath}
\mathbf{E}(T)=\log_{2}\bigg(\frac{1}{\epsilon}\bigg)+5.4+\varphi(\epsilon)\quad\text{where $|\varphi(\epsilon)|\leq 0.2$ when $\epsilon\to 0^{+}$.}
\end{displaymath}

Using batch generation, we propose a simple method in section \ref{sect_convolution} with expected complexity $-\log_{2}(\epsilon)+1.2353\ldots$ by using a convolutional method for the truncated exponential. An exponential random variable is the sum of a truncated exponential variable and of a geometric random variable. The mantissa part is distributed according to geometric random variable with parameter $1/e$ independent of the fractional part which is distributed as a truncated exponential variable over the interval $(0,1)$.

\subsubsection{The normal law}

For the normal law, Karney \cite{Kar_2016} proposes a method that addresses the variable approximation issue but does not offer explicit bounds. The inversion method yields an explicit upper bound of
\begin{displaymath}
\mathbf{E}(T)=\log_{2}\bigg(\frac{1}{\epsilon}\bigg)+\log_{2}\sqrt{2\pi{}e}+o(1),
\end{displaymath}
but the drawback is that this requires the presence of (an oracle for) $F^{-1}$, the inverse gaussian distribution function. Even the partition method requires a nontrivial oracle, namely $F$. To sidestep this, one can use a slightly more expensive method based on the Box-M\"{u}ller \cite{BoxMul_1958}, which states that the pair of random variables $\big(\sqrt{2E}V_{1},\sqrt{2E}V_{2}\big)$ with $E$ exponential and $(V_1,V_2)$ uniform on the unit circle, provides a standard gaussian in $\mathbb{R}^{2}$ of zero mean and unit covariance matrix. The random variable $\sqrt{2E}$ is a Maxwell random variable, that is, it has density $re^{-r^{2}\slash 2}=f_{M}(r)$ for $r>0$. The differential entropy of $f_{M}$ is given by
\begin{align*}
H(f_{M})&=\int_{0}^{\infty}{f_{M}(r)\log_{2}\bigg(\frac{1}{f_{M}(r)}\bigg)dr}=\frac{1}{\log2}\int_{0}^{\infty}{f_{M}(r)\bigg(\log\bigg(\frac{1}{r}\bigg)+\frac{r^{2}}{2}\bigg)dr}\\
&=\frac{1}{\log2}\bigg(\frac{1}{2}\bigg(\gamma-\log(2)\bigg)+1\bigg)=1.359068\ldots,
\end{align*}
where $\gamma=0.577215\ldots$ is the Euler-Mascheroni constant.

We sketch the procedure, which also serves as an example for more complicated random variate generation problems. We choose $d=2$ and $p=\infty$ so that the accuracy for two independent normals is $\epsilon$. We first generate a Maxwell random variable $M$ by inversion since $F_{M}(r)=1-e^{-\frac{r^{2}}{2}}$ and $F_{M}^{-1}(u)=\sqrt{-2\log(1-u)}$. The precision needed for $M$ is $\frac{\epsilon}{2}$. The Maxwell law is unimodal with mode at $r=1$. Its left piece has probability $1-\frac{1}{\sqrt{e}}$. So we pick a piece randomly using on average no more than two bits, and then we apply inversion on the appropriate piece. By theorem \ref{thm_inv_mono_bornee}, we use $T_1$ random bits where
\begin{displaymath}
\mathbf{E}(T_1)\leq \log_{2}\bigg(\frac{2}{\epsilon}\bigg)+H(f_{M})+2+o(1).
\end{displaymath}
The generated approximation is called $M'$.

Second, we generate a uniform random variable $U\in[0,2\pi)$ with accuracy $\frac{\epsilon\slash 2}{M'+\big(\epsilon\slash 2\big)}$. The generated value $U'\in[0,2\pi)$ has $|U-U'|\leq \frac{\epsilon\slash 2}{M'+(\epsilon\slash 2)}$. Since $U$ has differential entropy $\log_{2}(2\pi)$, and by using the Lebesgue's theorem on dominated convergence, we have that
\begin{align*}
\mathbf{E}(T_{2})&\leq \mathbf{E}\bigg(\log_{2}\bigg(\frac{M'+\big(\epsilon\slash 2\big)}{\epsilon\slash 2}\bigg)\bigg)+\log_{2}(2\pi)+o(1)\\
&\leq \mathbf{E}\bigg(\log_{2}\bigg(\frac{M+\big(\epsilon\slash 2\big)}{\epsilon\slash 2}\bigg)\bigg)+\log_{2}(2\pi)+o(1)\\
&=\log_{2}\bigg(\frac{2}{\epsilon}\bigg)+\mathbf{E}\big(\log_{2}(M)\big)+\log_{2}(2\pi)+o(1).
\end{align*}
Finally we then return $\big(M'\sin(U'), M'\cos(U')\big)$ for which both $|M'\sin(U')-M\sin(U)|\leq \epsilon$ and $|M'\cos(U')-M\cos(U)|\leq \epsilon$ holds simultaneously. Indeed, we observe that
\begin{displaymath}
|\sin(U')-\sin(U)|\leq |U-U'|\leq \frac{\epsilon\slash 2}{M'+\big(\epsilon\slash 2\big)}
\end{displaymath}
and similarly for the cosine. Next, we have that
\begin{align*}
|M'\sin(U')-M\sin(U)|&\leq |M'-M||\sin(U')|+M|\sin(U')-\sin(U)|\\
&\leq |M'-M|+M|U'-U|\\
&\leq \frac{\epsilon}{2}+M\frac{\epsilon\slash 2}{M'+\big(\epsilon\slash 2\big)}\leq \epsilon.
\end{align*}
Combining everything together, the total expected number of bits is not more than
\begin{align}
\mathbf{E}(T_1)+\mathbf{E}(T_2)&\leq 2\log_{2}\bigg(\frac{2}{\epsilon}\bigg)+\mathbf{E}\big(\log_{2}(M)\big)+H(f_{M})+2+\log_{2}(2\pi)+o(1)\nonumber\\
&=2\log_{2}\bigg(\frac{1}{\epsilon}\bigg)+2+\log_{2}(2\pi{}e)+o(1)\nonumber\\
&=2\log_{2}\bigg(\frac{1}{\epsilon}\bigg)+ 6.094191\ldots +o(1)\label{peanut_buster}.
\end{align}
We point out that the difference between (\ref{peanut_buster}) and the lower bound to generate two independent standard normal random variables is less that $4+o(1)$.

\subsection{Bisection}\label{sect_bisection}

Bisection is another the building block for continuous random variate generation. We develop a flavour of bisection that is convenient for our needs as exposed in section \ref{sect_vN}. Suppose we have a continuous distribution $F$ over a compact interval $[a,b]$ for $b>a$. A bisection method is a method that halves the support of $F$. A bisection method yields very naturally a full binary tree, that is, one in which all internal nodes have two children. In comparison to the inversion method from section \ref{sect_inversion}, the distribution is halved at every call to \texttt{RandomBit} and for which the resulting binary tree is not necessarily full.

We present algorithm \ref{bisecto_algo} from Devroye and Gravel \cite{DevGra_2017} which assumes an access to both $F$ and $F^{-1}$. The tree structure beneath algorithm \ref{bisecto_algo} is such that each internal node corresponds to a subinterval of $[a,b]$ of length greater than $2\epsilon$, the root represents the original interval $[a,b]$ of volume $|b-a|$, and leaves represent intervals of length less than or equal to $2\epsilon$.

\begin{breakablealgorithm}
\caption{A bisection algorithm (Devroye and Gravel \cite{DevGra_2017})}\label{bisecto_algo}
\begin{algorithmic}[1]
\Require Algorithms to compute $F$ and $F^{-1}$
\Require $a$, $b>a$, and $\epsilon>0$
\Ensure $X_{\epsilon}$ such that $|X-X_{\epsilon}|<\epsilon$
\State $I\leftarrow [a,b]$
\State $J\leftarrow [F(a),F(b)]=[0,1]$
\Loop
\If {$|I|= b-a\leq 2\epsilon$}
\State $X_{\epsilon}\leftarrow \frac{a+b}{2}$
\State{\textbf{Return} $X_{\epsilon}$} \Comment{Exit}
\Else
\State $B\leftarrow \texttt{RandomBit}$
\State $z\leftarrow F^{-1}\Big(\frac{F(a)+F(b)}{2}\Big)$
\If {$B=0$}
\State $I\leftarrow[a,z]$
\State $J\leftarrow\Big[F(a),\frac{F(a)+F(b)}{2}\Big]=\big[F(a),F(z)\big]$
\Else\Comment{$B=1$}
\State $I\leftarrow[z,b]$
\State $J\leftarrow\Big[\frac{F(a)+F(b)}{2}, F(b)\Big]=\big[F(z),F(b)\big]$
\EndIf
\EndIf
\EndLoop
\end{algorithmic}
\end{breakablealgorithm}

\begin{theorem}\label{thm_bisect}
Let $X$ be a continuous distribution with distribution $F:[a,b]\to [0,1]$ and let $T$ be the number of calls to \texttt{RandomBit} in algorithm \ref{bisecto_algo}.
\begin{enumerate}
\item[(i)] Upon returning the center of the halting interval, that is $X_{\epsilon}$, the random variables $X$ and $X_{\epsilon}$ are coupled such that $|X-X_{\epsilon}|\leq \epsilon$.
\item[(ii)] The expected complexity $\mathbf{E}(T)$ is such that
\begin{equation}
\mathbf{E}(T)\leq 3+\log_{2}^{+}\bigg(\frac{b-a}{2\epsilon}\bigg).\label{caca_boulettes}
\end{equation}
\end{enumerate}
\end{theorem}

Before the proof of theorem \ref{thm_bisect}, we observe that the bound (\ref{caca_boulettes}) cannot be improved in general by more than $3$ bits. Indeed, we just consider the uniform distribution on $[a,b]$. Since all intervals have length to $\frac{b-a}{2^{i}}$ after $i$ calls to \texttt{RandomBit}, we have
\begin{displaymath}
T=\min\bigg\{i\geq 0\colon\frac{b-a}{2^{i}}\leq 2\epsilon\bigg\}=\max\bigg\{0,\bigg\lceil\log_{2}\frac{b-a}{2\epsilon}\bigg\rceil\bigg\}.
\end{displaymath}

\begin{proof}[Proof of theorem \ref{thm_bisect}]
To show part (i), as pointed out before, the bisection method yields a full binary tree. Each internal node corresponds to a subinterval of $[a,b]$ of length greater than $2\epsilon$, the root represents the original interval $[a,b]$ of length $L$, and halting-leaves represent intervals of length less than or equal to $2\epsilon$.

At every call of \texttt{RandomBit}, the random binary choice picks either $[a,z]$ or $[z,b]$ as an interval with probability $\frac{1}{2}$. Define $X$ as the limit when the algorithm runs without halting. Upon exit, $X_{\epsilon}$ defined to be the midpoint of an interval of length at most $2\epsilon$ that also contains $X$, we must have $|X-X_{\epsilon}|\leq\epsilon$. This shows part (i).

To prove part (ii), denote by $\mathcal{L}$ and by $\mathcal{I}$, the set of leaves and the set of internal nodes of the underlying full binary tree, respectively. The depth of a node $u\in \mathcal{I}\cup\mathcal{L}$ is denoted by $\texttt{d}(u)$. It is of course possible that $\mathcal{I}$ and $\mathcal{L}$ are both infinite. Because the leaves form a non-overlapping covering of $[a,b]$, the random walk produced by the algorithm \ref{bisecto_algo} always stops. For all possible random walks, we have that
\begin{equation}
\sum_{u\in\mathcal{L}}{\frac{1}{2^{\mathtt{d}(u)}}}\leq 1\quad\text{and}\quad\mathbf{E}(T)=\sum_{u\in\mathcal{L}}{\frac{\mathtt{d}(u)}{2^{\mathtt{d}(u)}}}.\label{citozine1}
\end{equation}
The inequality from (\ref{citozine1}) follows also from Kraft's inequality.

We show a chain of inequalities in order to complete the proof, but we introduce some notation before. We denote by $N_{\ell}$ the number of internal nodes at depth $\ell$ in the tree, that is,
\begin{displaymath}
N_{\ell}=\sum_{v\in\mathcal{I}}{\mathds{1}\{\mathtt{d}(v)=\ell\}}\quad\text{for $\ell\geq 0$}.
\end{displaymath}
We denote also by $A(u)$ the set of ancestors of $u$ and by $D(v)$ the set of descendants of $v$. For any node $u$, we have $\mathtt{d}(u)=\sum{\mathds{1}\{v\in A(u)\setminus\{u\}\}}$. We then deduce the following inequalities:
\begin{align*}
&\mathbf{E}(T)=\sum_{u\in\mathcal{L}}\mathop{\sum_{v \in A(u)}}_{v\neq u}{\frac{1}{2^{\mathtt{d}(v)}}\frac{1}{2^{\mathtt{d}(u)-\mathtt{d}(v)}}}=\sum_{v\in\mathcal{I}}{\frac{1}{2^{\mathtt{d}(v)}}}\mathop{\sum_{u\in D(v)}}_{u\in\mathcal{L}}{\frac{1}{2^{\mathtt{d}(u)-\mathtt{d}(v)}}}\\
&\qquad\quad\leq \sum_{v\in\mathcal{I}}{\frac{1}{2^{\mathtt{d}(v)}}}=\sum_{\ell=0}^{\infty}{\frac{N_{\ell}}{2^{\ell}}}\leq \sum_{\ell=0}^{\infty}{\frac{b-a}{2^{\ell}}}
\end{align*}
The last inequality follows from the fact that, at depth $\ell$, all intervals associated with nodes are disjoint and each interval node corresponds to an interval strictly larger than $2\epsilon$. Also because we have a binary tree, we necessarily have that $N_{\ell}\leq 2^{\ell}$ from which $N_{\ell}\leq \min\big\{\big\lfloor\frac{b-a}{2\epsilon}\big\rfloor,2^{\ell}\big\}$. Let $\ell_0>0$ be the threshold depth for which $b-a<2\epsilon$, that is
\begin{displaymath}
\ell_{0}=\max\bigg\{0,\bigg\lceil\log_{2}\frac{b-a}{2\epsilon}\bigg\rceil\bigg\}.
\end{displaymath}
The proof is completed because
\begin{displaymath}
\sum_{\ell=0}^{\infty}{\frac{N_{\ell}}{2^{\ell}}}\leq\sum_{\ell=0}^{\ell_{0}}{1}+\sum_{\ell=\ell_{0}+1}^{\infty}{\bigg\lfloor\frac{b-a}{2\epsilon}\bigg\rfloor\frac{1}{2^{\ell}}}=\ell_{0}+1+\bigg\lfloor\frac{b-a}{2\epsilon}\bigg\rfloor\frac{1}{2^{\ell_0}}\\
\leq\log_{2}^{+}\bigg(\frac{b-a}{2\epsilon}\bigg)+3.
\end{displaymath}
\end{proof}

\subsection{Von Neumann's sampling algorithm extended}\label{sect_vN}

This section discusses Von Neumann's \cite{vN_1951} sampling method which we extend to: the context of multiple precision arithmetic, and the fact that have an access to a discrete source of unbiased i.i.d.\ random bits. The work of this section can be found in Devroye and Gravel \cite{DevGra_2017}. In this section, distributions are necessarily absolutely continuous. Given two densities $f$ and $g$ with identical support and such that $f(x)\leq Cg(x)$ for some constant $C>1$, Von Neumann's original idea is to test whether $UCg(X)\leq f(X)$ or not given that $U$ is uniformly distributed on the interval $(0,1)$. We sometimes refer to $f$ as the targeted density and to $g$ as the easy density. We observe that one must have a way to sample $g$ and hence the qualifier easy. Also when $f$ has a compact support, the easy density can be the uniform density in which case the cutoff inequality is $f(x)\leq C$. In theory, the test supposed the capability to compare random quantities with infinite precision and store unbounded quantities such as $U$ and $X$. Current classical algorithms approximate the quantities involved in the evaluation of the test. As more and more random bits are obtained through $\texttt{RandomBit}$, more and more accurate quantities can be computed until a decision can be made exactly without computational error for the inequality test, and up to the desired accuracy for the random output.

We split our extension in two cases: the compact support case in section \ref{sect_vN_compact} and the non-compact support case in section \ref{sect_vn_non_compact}. In both cases, we assume the ability to compute infima and suprema or their ratios over given compact subset of the support. This is a very reasonable assumption given the current state of the art in terms of libraries and software to compute with multiple precision as briefly mentioned in section \ref{sect_concrete_imple}. We use a quadtree as a data structure to represent the partition of the space under the test $UCg(X)\leq f(X)$. Our extension is guaranteed to deliver an output with the requested desired accuracy. Our approach assumes that $f$ is Riemann-integrable. We derive the expected complexity of the number of random bits and indeed observe that it is near the universal lower bounds from section \ref{sect_cts_gen}. We may use the noun complements Von Neumann or rejection interchangeably when referring to algorithms or methods in this section.

\subsubsection{An algorithm for densities with compact support}\label{sect_vN_compact}

In this section, we assume that $f$ is Riemann-integrable and supported on $[0,1]^{d}$ which is equivalent to the assumption that $f$ is almost-everywhere continuous, bounded, and supported on $[0,1]^{d}$. Our algorithm requires a method, denoted by $\mathtt{M}$, such that on input $R\subseteq[0,1]^{d}$ computes
\begin{equation}
\inf\big\{f(x)\colon x\in R\big\}\quad\text{and}\quad\sup\big\{f(x)\colon x\in R\big\}\quad\text{for any $R\subseteq\mathbb{R}^{d}$.}\label{inf_sup_machin}
\end{equation}
If $R=\{x\}$, the method \texttt{M} returns $f(x)$. Without the possibility to compute quantities from (\ref{inf_sup_machin}), sampling using the rejection method seems impossible in total generality. Without loss of generality, the description of $f$ can be hardcoded in \texttt{M}. An invocation of \texttt{M} returns $C=\sup\{f(x)\colon x\in[0,1]^{d}\}$ which is a finite number by assumption since Riemann-integrable functions are bounded by definition. At once, we have a simple bound for applying the original Von Neumann sampling algorithm \ref{orig_vnabc} given as follow, and which is subsequently modified into algorithm \ref{algo_reject_compact} to match with more practical and realistic scenarios.

\begin{breakablealgorithm}
\caption{Von Neumann's original rejection algorithm}\label{orig_vnabc}
\begin{algorithmic}[1]
\Loop
\State Generate $X$ uniformly on $[0,1]^{d}$.
\State Generate $U$ uniformly on $[0,1]$.
\If{$UC\leq f(X)$}
\State{\textbf{Return} $X$}
\EndIf
\EndLoop
\end{algorithmic}
\end{breakablealgorithm}

Since we cannot generate $X$ and $U$ with infinite precision, at least two modifications are needed. One modification is to take into account the precision $\epsilon$ desired for $X$, and the other modification is to take into account the discreteness of the random source. We consider the rectangle $[0,1]^{d}\times [0,C]$, denoted by $R_0$,  and its $2^{d+1}$ sub-rectangles defined by the $2^{d+1}$ quadrants centered at $\big(\frac{1}{2},\frac{1}{2},\ldots,\frac{1}{2},\frac{C}{2}\big)$. In the data structure literature, the latter partition, when applied recursively, leads to a quadtree as in Samet \cite{Sam_book_2006} for instance. Let us denote by $Q$ the infinite size quadtree obtained by recursively refining $R_0$. A rectangle is recursively split around its center point, and so forth, as illustrated by figures \ref{fig_2dim_decomp_for_quadtree} and \ref{fig_quadtree_decomp_example}.

\begin{figure}
\begin{centering}
\includegraphics[scale=0.75]{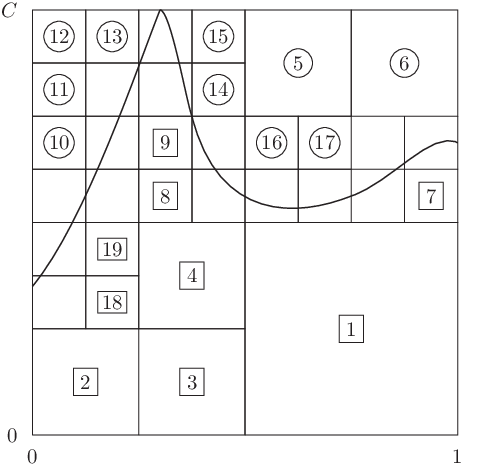}
\caption{A possible decomposition of $[0,1]\times [0,C]$ after $4$ divisions.}\label{fig_2dim_decomp_for_quadtree}
\end{centering}
\end{figure}

\begin{figure}
\begin{centering}
\includegraphics[scale=0.75]{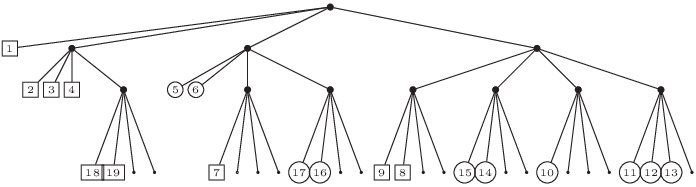}
\caption{Quadtree for decomposition on figure \ref{fig_2dim_decomp_for_quadtree}.}\label{fig_quadtree_decomp_example}
\end{centering}
\end{figure}

In Von Neumann's algorithm, to decide if $UC\leq f(X)$ for $(X,U)\in R_{0}$ is equivalent to find a rectangle $R$ in the quadtree $Q$ with the property that either
\begin{align*}
&R\subseteq \big\{(x,y)\in R_{0}\colon y\leq f(x)\big\}\quad\text{(we accept since $UC\leq f(X)$)}\\
\text{or}&\\
&R\subseteq \big\{(x,y)\in R_{0}\colon y> f(x)\big\}\quad\text{(we reject since $UC> f(X)$)}.
\end{align*}
However, the former and the latter must be done carefully without overlapping rectangles so as to trim $Q$ such that leaves are associated to halting rectangles. Thus we have
\begin{align}
&\big\{(x,y)\in R_{0}\colon y\leq f(x)\big\}=\bigcup\{R\colon\text{$R$ is an accepting rectangle}\}\label{accept_rect}\\
\text{and}&\nonumber\\
&\big\{(x,y)\in R_{0}\colon y> f(x)\big\}=\bigcup\{R\colon\text{$R$ is a rejecting rectangle}\}.\label{reject_rect}
\end{align}

Below, we will see that Riemann-integrability of $f$ suffices for the decomposition given by (\ref{accept_rect}) and (\ref{reject_rect}). If a rejecting rectangle is found, then the procedure is repeated. If an accepting rectangle is found, say $\prod_{i=1}^{d}{[a_i,b_i]}\times [\alpha,\beta]$ with $0<b_i-a_i<1$, then it suffices to generate an $\epsilon$-accurate uniform random variables over the projection $\prod_{i=1}^{d}{[a_i,b_i]}$ by using bisection, as presented in section \ref{sect_bisection}, over each dimension. We recall from theorem \ref{thm_bisect} that the expected complexity for the bisection part of the algorithm is
\begin{displaymath}
3+\sum_{i=1}^{d}{\log_{2}^{+}\bigg(\frac{b_i-a_i}{\epsilon}\bigg)}\leq 3+d\log_{2}^{+}\bigg(\frac{1}{\epsilon}\bigg).
\end{displaymath}
We note that the checks
\begin{align}
&R\subseteq \{(x,y)\in R_0\colon f(x)\leq y\}\label{check_one}\\
\text{and}&\nonumber\\
&R\subseteq \{(x,y)\in R_0\colon f(x) > y\}\label{check_two}
\end{align}
can be carried out using method \texttt{M}. For convenience, let us write
\begin{align*}
&f^{+}=\sup\big\{f(x)\hspace{1pt}\colon (x,y)\in R\text{ for some $y$}\big\},\hspace{2pt} y^{+}=\sup\big\{y\hspace{1pt}\colon (x,y)\in R\text{ for some $x$}\big\},\hspace{2pt}\\
&f^{-}=\inf\big\{f(x)\hspace{1pt}\colon(x,y)\in R\text{ for some $y$}\big\},\hspace{4pt}
y^{-}=\inf\big\{y\hspace{1pt}\colon(x,y)\in R\text{ for some $x$}\big\}.
\end{align*}
Then (\ref{check_one}) holds if $f^{+}\leq y^{-}$, and (\ref{check_two}) holds if $f^{-}\geq y^{+}$.


\begin{breakablealgorithm}
\caption{Extension of algorithm \ref{orig_vnabc}}\label{algo_reject_compact}
\begin{algorithmic}[1]
\Require{\texttt{M}, $\epsilon$}\Comment{Note that \texttt{M} on input $x$ returns $f(x)$.}
\Ensure{$X_{\epsilon}$ such that $\|X-X_{\epsilon}\|_{\infty}<\epsilon$}\Comment{$X$ has density $f$.}
\State{$C\leftarrow\sup\{f(x)\colon x\in[0,1]^{d}\}$}\Comment{By using \texttt{M}}
\State{$R\leftarrow R_0=[0,1]^{d}\times\big[0,C\big]$}\label{init_line_reject_algo}
\State{$\textrm{Decision}\leftarrow \textrm{None}$}
\Repeat
\State{$f^{-}\leftarrow\inf\{f(x)\colon (x,y)\in R\quad\text{for some $y$}\}$}\Comment{By using \texttt{M}}
\State{$f^{+}\leftarrow\sup\{f(x)\colon (x,y)\in R\quad\text{for some $y$}\}$}\Comment{By using \texttt{M}}
\State{$y^{-}\leftarrow\inf\{y\colon (x,y)\in R\quad\text{for some $x$}\big\}$}
\State{$y^{+}\leftarrow\sup\{y\colon (x,y)\in R\quad\text{for some $x$}\big\}$}
\If{$f^{+}\leq y^{-}$}\Comment{$R\subseteq \{(x,y)\in R_0 \colon f(x)\leq y\}$}
\State{$\textrm{Decision}\leftarrow \textrm{Accept}$}
\ElsIf{$f^{-}\geq y^{+}$}\Comment{$R\subseteq \{(x,y)\in R_0 \colon f(x)> y\}$}
\State{$\textrm{Decision}\leftarrow \textrm{Reject}$}
\Else\Comment{Update $R$}
\State{$z\leftarrow \text{center of $R$}$}
\State{Select uniformly one vertex $v$ of $R$ among its $2^{d+1}$ vertices.}\Comment{This requires $d+1$ calls to \texttt{RandomBit}.}
\State{Update $R$ by selecting the unique rectangle containing the line segment made from $v$ and $z$.}
\EndIf
\Until{$\textrm{Decision} \neq \textrm{None}$}
\If{$\textrm{Decision}=\textrm{Reject}$}
\State{Goto line (\ref{init_line_reject_algo})}\Comment{Restart the algorithm.}
\Else
\State{$R^{*}\leftarrow \{x\colon (x,y)\in R\quad\text{for some $y$}\}$}\Comment{$R^{*}$ is the projection of $R$ onto its first $d$ coordinates.}
\State{Use bisection to generate $X_{\epsilon}$.}\label{bisecto_sub_routine_line_reject_algo}\Comment{Apply algorithm \ref{bisecto_algo} on every one dimensional subspace of $R^{*}$.}
\State{\textbf{Return} $X_{\epsilon}$}
\EndIf
\end{algorithmic}
\end{breakablealgorithm}

\begin{theorem}\label{thm_quadtree}
Let $X$ be a continuous random variable on $[0,1]^{d}$ with a Riemann-integrable density $f$. Algorithm \ref{algo_reject_compact} halts with probability one and outputs $X_{\epsilon}$ such that $\|X-X_{\epsilon}\|_{\infty}<\epsilon$.
\end{theorem}
Theorem \ref{thm_quadtree} says that algorithm \ref{algo_reject_compact} is correct. We analyze the expected complexity of algorithm \ref{algo_reject_compact} more in details after the following proof. However we point out that in order to show that algorithm \ref{algo_reject_compact} halts, we end up to obtain indirectly its complexity excluding the bisection part of it.
\begin{proof}[Proof of theorem \ref{thm_quadtree}]
Initially, the hyper-rectangle is $[0,1]^{d}\times [0,\sup{f}]$. The quadtree underlying algorithm \ref{algo_reject_compact} allows to partition the initial hyper-rectangle into a collection of smaller hyper-rectangles for which either (\ref{accept_rect}) or (\ref{reject_rect}) holds. Let $T$ be the number of iterations of the algorithm before halting. In other words, $T$ is the depth of the leaf reached upon halting by randomly walking down the quadtree. We show that $\lim_{k\to\infty}{\mathbf{P}\{T>k\}}=0$
and therefore that either (\ref{accept_rect}) or (\ref{reject_rect}) holds. The partition is made of $2^{(d+1)k}$ hyper-rectangles, each of Lebesgue measure $\squash{\frac{1}{2^{k}}}\squash{\frac{1}{2^{k}}}\cdots\squash{\frac{1}{2^{k}}}\squash{\frac{C}{2^{k}}}$. Let $N_{k}$ be the number of cells in the partition for which we cannot decide, that is, for which
\begin{align*}
&\sup\{f(x)\colon (x,y)\in R\text{ for some $y$}\}\geq \inf\{y\colon (x,y)\in R\text{ for some $x$}\}\\
\text{and}&\\
&\inf\{f(x)\colon (x,y)\in R\text{ for some $x$}\}\leq \sup\{y\colon (x,y)\in R\text{ for some $x$}\}.
\end{align*}
Then we have that
\begin{equation}
\mathbf{P}\{T>k\}=\frac{N_k}{2^{(d+1)k}}.\label{von_Neumann_tail}
\end{equation}
For every rectangle $R$, let us write $R^{\star}$ as the projection of $R$ onto $\mathbb{R}^{d}$, that is, $R^{\star}=\{x\colon (x,y)\in R\text{ for some $y$}\}$. For a fixed one-dimensional subspace, group the $2^{k}$ cells $R^{\star}$ with the same projection (equivalence classes), and verify that, among these $2^{k}$ cells, the number of cells that intersect the graph of $f$ is at most
\begin{displaymath}
\bigg(\frac{\sup\{f(x)\colon x\in R^{\star}\}-\inf\{f(x)\colon x\in R^{\star}\}}{C}\bigg)2^{k}+2.
\end{displaymath}
Let us write $\mathcal{P}_{k}^{\star}$ as the collection of all projections $R^{\star}$ after $k$ iterations. Since there are $2^{dk}$ rectangles $R^{\star}$, we have that
\begin{equation}
N_{k}\leq \sum_{R^{\star}\in\mathcal{P}_{k}^{\star}}{\bigg(\bigg(\frac{\sup\{f(x)\colon x\in R^{\star}\}-\inf\{f(x)\colon x\in R^{\star}\}}{C}\bigg)2^{k}+2\bigg)}.\label{upbnd_Nk}
\end{equation}
Consider the Riemann approximations, $I^{-}$ and $I^{+}$, for the integral of $f$ which are given by
\begin{equation}
I_{k}^{-}=\sum_{R^{\star}\in\mathcal{P}_{k}^{\star}}{\inf\{f(x)\colon x\in R^{\star}\}\lambda(R^{\star})}\text{  and  } I_{k}^{+}=\sum_{R^{\star}\in\mathcal{P}_{k}^{\star}}{\sup\{f(x)\colon x\in R^{\star}\}\lambda(R^{\star})},\label{riemann_approx}
\end{equation}
where $\lambda(R^{\star})=\frac{1}{2^{dk}}$. Combining (\ref{von_Neumann_tail}), (\ref{upbnd_Nk}) and (\ref{riemann_approx}), we finally obtain that
\begin{equation}
\mathbf{P}\{T>k\}\leq \frac{2}{2^{k}}+\frac{I_{k}^{+}-I_{k}^{-}}{C},\label{non_halting_prob_vn}
\end{equation}
and by the Riemann integrability of $f$, the latter (\ref{non_halting_prob_vn}) tends to $0$ as $k\to\infty$ which implies that the algorithm halts with either accepting or rejecting the rectangle.

The last part of the theorem which states that $X_{\epsilon}$ is such that $\|X-X_{\epsilon}\|_{\infty}<\epsilon$ follows from the correctness of algorithm \ref{bisecto_algo} in section \ref{sect_bisection}.
\end{proof}

\begin{remark}
For the complexity of algorithm \ref{algo_reject_compact}, we need to consider the total number of trials before deciding that we denote by $N$ here. Let $T_i$ is the number of iterations in the $i$-th trial. The random variables $T_i$ are clearly i.i.d.\ The number of random bits used is $(d+1)\sum_{i=1}^{N}{T_i}$. Since $\mathbf{E}(N)=C=\sup\{f(x)\colon x\in R_0\}$, the expected number of random bits, excluding the bisection phase, is $C(d+1)\mathbf{E}(T_1)$.
\end{remark}

We call $f$ a monotone density on $[0,1]^{d}$ if it decreases along at least one of the dimensions, that is, there exists $i\in\{1,\ldots,d\}$ such that for all vectors $(x_1,\ldots,x_i,\ldots,x_d)$, $(x_1,\ldots,x'_i,\ldots,x_d)\in[0,1]^{d}$ such that $x_i\leq x'_i$, then $f(x_1,\ldots,x_i,\ldots,x_d)\geq f(x_1,\ldots,x'_i,\ldots,x_d)$.
\begin{corollary}\label{cor_mono_dense}
Let $f$ be a monotone density, and let $T$ as before, then we have
\begin{displaymath}
\mathbf{P}\{T>k\}=\frac{N_k}{2^{(d+1)k}}\leq\frac{2}{2^{k}}\quad\text{for $k\geq 0$, and thus}\quad\mathbf{E}(T)=\sum_{k=0}^{\infty}{\mathbf{P}\{T>k\}}\leq 4.
\end{displaymath}
In other words, for monotone densities, the inner loop of algorithm \ref{algo_reject_compact} has a guaranteed uniform performance.
\end{corollary}
\begin{proof}[Proof of corollary \ref{cor_mono_dense}]
As before, let $N_{k}$ be the number of cells at level $k$ that are visited by $f$, that is, that are intersecting the graph of $f$. Therefore we have $N_{k}\leq 2\cdot2^{k}\cdot2^{(d-1)k}$ because the domain of $f$ is divided into $2^{dk}$ cells and the $2^{k}$ cells along the $i^{\text{th}}$ dimension yields a walk which is at most of length $2\cdot2^{k}$ as illustrated on figure \ref{fig_mono_curve}.
\end{proof}

\begin{figure}
\begin{centering}
\includegraphics[scale=0.75]{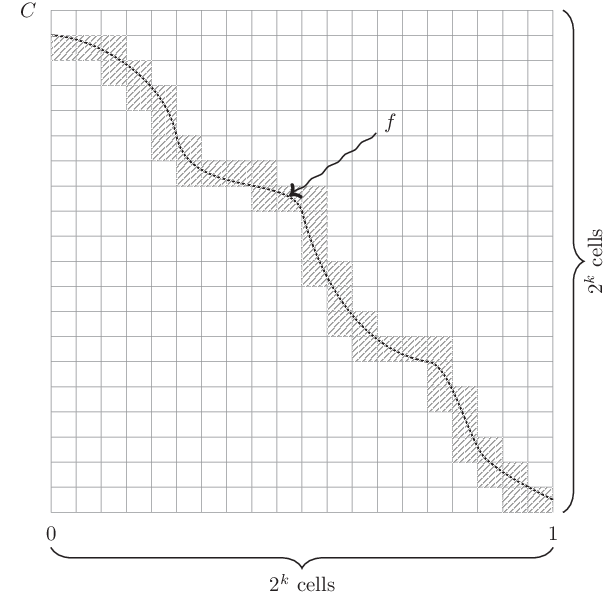}
\caption{The number of cells visited by the monotone curve $f$ is at most $2\cdot2^{k}$ cells.}\label{fig_mono_curve}
\end{centering}
\end{figure}
As noted earlier, for any coordinate-wise monotone density on $[0,1]^{d}$, we have $\mathbf{E}(T)\leq 4$. However, for general Riemann-integrable densities we cannot insure that $\mathbf{E}(T)$ converges, and we address this point hereafter.
\begin{theorem}\label{thm_compact_case}
Let $f$ be a Riemann-integrable probability density function defined over $[0,1]^{d}$ and let $C$ be such that $C=\sup\{f(x)\colon x\in[0,1]^{d}\}$.
\begin{enumerate}
\item If $f$ is monotone in at least one coordinate, then the expected number of perfect coin flips required to generate an $\epsilon$-accurate approximation is not more than $4C(d+1)+3+d\log_{2}^{+}\big(\frac{1}{2\epsilon}\big)$.
\item If $I_{k}^{+}$ and $I_{k}^{-}$ are the Riemann approximations for an equally spaced partition with parts of size $2^{dk}$, that is, each coordinate is split into $2^{k}$ equal intervals, then the expected number of bits needed to generate an $\epsilon$-approximation is not more than $4C(d+1)+(d+1)\Delta(f)+3+d\log_{2}^{+}\big(\frac{1}{2\epsilon}\big)$ where $\Delta(f)= \sum_{k=0}^{\infty}{\big(I_{k}^{+}-I_{k}^{-}\big)}$.
\end{enumerate}
\end{theorem}

\begin{proof}[Proof of theorem \ref{thm_compact_case}]
Just recall the estimates of $\mathbf{E}(T)$ obtained above and recall the upper bound $\mathbf{P}\{T>k\}$ in terms of $I_{k}^{+}-I_{k}^{-}$.
\end{proof}

\begin{remark}
Part (2) of the theorem \ref{thm_compact_case} is useful only if $\Delta(f)<\infty$. For most densities, we have $\Delta(f)<\infty$, and simple sufficient condition is obtained using an argument involving the modulus of continuity as in \cite{Whi_Wat_1996} defined by
\begin{displaymath}
\omega_{f}(\delta)=\sup\{|f(x)-f(y)|\colon \|x-y\|_{\infty}\leq\delta\}\quad\text{for $\delta>0$}.
\end{displaymath}
We observe that
\begin{align*}
&\sum_{k=0}^{\infty}{\omega_{f}\bigg(\frac{\sqrt{d}}{2^{k}}\bigg)}<\infty\quad\Longrightarrow\quad \Delta(f)<\infty,\\
\text{because}\\
&I_{k}^{+}-I_{k}^{-}=\frac{1}{2^{dk}}\sum_{R^{\star}\in\mathcal{P}_{k}^{\star}}{\big(\sup\{f(x)\colon x\in R^{\star}\}-\inf\{f(x)\colon x\in R^{\star}\}\big)}\\
&\quad\leq \frac{1}{2^{dk}}\sum_{R^{\star}\in\mathcal{P}_{k}^{\star}}{\sup\{|f(x)-f(x')|\colon x,x'\in R^{\star}\}}\leq \omega_{f}\bigg(\frac{\sqrt{d}}{2^{k}}\bigg).
\end{align*}
To conclude this remark, it suffices that as $\delta\to 0^{+}$, $\omega(\delta)=O\big(\rfrac{1}{\log^{1+\alpha}(\delta^{-1})}\big)$ or $\omega(\delta)=O(\delta^{\alpha})$ for some $\alpha>0$.
\end{remark}

\begin{remark}
The expected number of calls to \texttt{RandomBit} behaves as $d\log_{2}^{+}\big(\frac{1}{\epsilon}\big)+O(1)$ as $\epsilon\to 0^{+}$, and thus matches the lower bound mentioned earlier for a Riemann-integrable density on a compact support.
\end{remark}

\subsubsection{An algorithm for densities with non-compact support}\label{sect_vn_non_compact}

Given a density $f:I\to [0,1]$ with $I\subseteq \mathbb{R}^{d}$, we use generally use the rejection method when we know a density $g:I\to[0,1]$ for which random variate generation is ``easy'' and for which we know a constant $C>1$ such that $\sup\big\{\frac{f(x)}{g(x)}\colon x\in\mathbb{R}\big\}=C<\infty$. The former affirmation is especially relevant when $I$ is not compact. When $I$ is compact, $g$ can be the uniform density most often. We recall Von Neumann's method in algorithm \ref{gen_algo}.

\begin{breakablealgorithm}
\caption{General rejection algorithm}\label{gen_algo}
\begin{algorithmic}
\Loop
\State Generate $X$ with density $g$
\State Generate $U$ uniformly on $[0,1]$
\If{$Cg(X)U<f(X)$}
\State{\textbf{Return} $X$}
\EndIf
\EndLoop
\end{algorithmic}
\end{breakablealgorithm}
We offer a generalization of algorithm \ref{gen_algo} under certain assumptions:
\begin{enumerate}
\item[1.] Assume for now that $d=1$, and that we can compute both $G$ and $G^{-1}$, where $G$ is the c.d.f.\ for $g$.
\item[2.] Assume furthermore that, for all $R\subset I\subseteq \mathbb{R}$, the subroutine \texttt{M} can compute $\sup\big\{\frac{f(x)}{g(x)}\colon x\in R\big\}$ and $\inf\big\{\frac{f(x)}{g(x)}\colon x\in R\big\}$.
\end{enumerate}
We observe that it is possible to do not modify \texttt{M} in (2) just above, but we opt to take the more convenient approach. We explain and analyze in the remainder of this section, that by a suitable transformation, it is necessary only to replace line (\ref{bisecto_sub_routine_line_reject_algo}) from algorithm \ref{sect_vN_compact} to obtain a rejection method for the non-compact case. We also detail in the remainder of this section the replacement in question which is algorithm \ref{algo_inv_bisect_for_vN_non_compact} below. Algorithm \ref{algo_inv_bisect_for_vN_non_compact} shares many features with the Han and Hoshi algorithm as we will see.

Define $C=\sup\big\{\frac{f(x)}{g(x)}\colon x \in I\big\}$, which is known thanks to \texttt{M}. As before in section \ref{sect_vN_compact}, the goal is to decompose the graph $\{(x,y)\colon y\leq f(x)\}$ into regions for which random variate generation is ``easy'', that is, where an instance of $g$ leads to acceptation. This can be mimicked by transforming the $x$-axis with $x\mapsto G(x)$ since $G$ is monotone and continuous. Using this transformation, we note that if $X$ has density $g$, then $G(X)$ is uniform on $[0,1]$. Furthermore, note that if $u=G(x)$, then
\begin{displaymath}
\frac{f\circ G^{-1}(u)}{g\circ G^{-1}(u)}=\frac{f(x)}{g(x)}= \tilde{f}(u)\quad\text{for $0\leq u\leq 1$,}
\end{displaymath}
where $\tilde{f}$ is a density on which we can use \texttt{M} properly modified as mentioned before. Since $\tilde{f}\leq C$, we can use a quadtree method similar to the one from algorithm \ref{algo_reject_compact} in order to select randomly a rectangle $R_i$ with probability $\lambda(R_i)$ from the decomposition
\begin{equation}
\big\{(u,v)\colon 0\leq u\leq 1\hspace{4pt}v\leq \tilde{f}(u)\big\}=\bigcup_{i\in\mathbb{N}}{\big\{R_{i}\colon\text{$R_{i}$ is an accepting rectangle}\big\}}.\label{decomp_vn_non_compact}
\end{equation}
We observe that if $\tilde{f}$ is Riemann integrable, then decomposition (\ref{decomp_vn_non_compact}) is valid so that we can decide whether to reject or accept with probability one. The expected number of coin flips required to decide is $(d+1)C\mathbf{E}(T)=2C\mathbf{E}(T)$ where $\mathbf{E}(T)\leq 4 + \sum_{k=0}^{\infty}{\big(I_{k}^{+}-I_{k}^{-}\big)}$, and this time, $I_{k}^{+}$ and $I_{k}^{-}$ are the Riemann approximations as of the integral of $\tilde{f}$ on an equally spaced partition of $[0,1]$. The quantity $\sum_{k=0}^{\infty}{\big(I_{k}^{+}-I_{k}^{-}\big)}$ is finite under smoothness conditions on $\tilde{f}$, and depends also on $C$, but clearly does not depend on $\epsilon$. We need therefore to analyze a method, shown below as algorithm \ref{algo_inv_bisect_for_vN_non_compact}, to generate an $\epsilon$-accurate outcome upon acceptance. Once a leaf that leads to acceptance is reached, say with label $[u_1,u_2]\times[v_1,v_2]=R$, algorithm \ref{algo_inv_bisect_for_vN_non_compact} generates an output $X_{\epsilon}$ such that $|X_{\epsilon}-G^{-1}(U)|\leq\epsilon$ and $(U,V)$ is uniform over $R$. We observe that, upon acceptance, $G^{-1}(U)$ has distribution function $G$ restricted to $[G^{-1}(u_1),G^{-1}(u_2)]$.

\begin{breakablealgorithm}
\caption{A modified Han and Hoshi's method for continuous distributions}\label{algo_inv_bisect_for_vN_non_compact}
\begin{algorithmic}[1]
\Require{$u_1$, $u_2$ such that $u_2>u_1$}\Comment{$u_2-u_1$ is the width of an accepting rectangle.}
\Require{$\epsilon>0$}
\Ensure{An $\epsilon$-accurate outcome}
\State{$x_1\leftarrow G^{-1}(u_1)$}
\State{$x_2\leftarrow G^{-1}(u_2)$}
\Loop
\If{$|x_2-x_1|\leq 2\epsilon$}
\State{$X_{\epsilon}\leftarrow\frac{x_1+x_2}{2}$}
\State{\textbf{Return} $X_{\epsilon}$}
\Else
\State{$\gamma\leftarrow \frac{u_1+u_2}{2}$}
\State{$B\leftarrow\texttt{RandomBit}$}
\If{$B=0$}
\State $u_2\leftarrow \gamma$
\State $x_2\leftarrow G^{-1}(u_2)$
\Else
\State $u_1\leftarrow \gamma$
\State $x_1\leftarrow G^{-1}(u_1)$
\EndIf
\EndIf
\EndLoop
\end{algorithmic}
\end{breakablealgorithm}

\begin{remark}
Algorithm \ref{algo_inv_bisect_for_vN_non_compact} is valid for singular distributions as well.
\end{remark}

\begin{figure}
\begin{centering}
\includegraphics{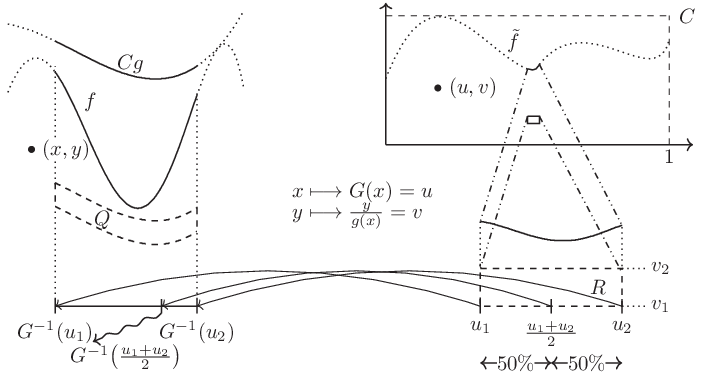}
\caption{An accepting uniform random rectangle $R$ and its pullback $Q$ for which $(u,v)\in R$ if and only if $(x,y)\in Q$.}\label{fig_bisect}
\end{centering}
\end{figure}

\begin{remark}
Note that with $x_1=G^{-1}(u_1)$ and $x_2=G^{-1}(u_2)$, we have
\begin{displaymath}
\lambda(R)=(u_2-u_1)(v_2-v_1)=\int_{u_1}^{u_2}{(v_2-v_1){d}u}=\int_{x_1}^{x_2}{(v_2-v_1)(g(x){d}x)}=\lambda(Q).
\end{displaymath}
Also if $(u,v)$ is such that $v<\tilde{f}(u)$, then the corresponding $(x,y)$ point is such that
\begin{displaymath}
y=vg(x)<\tilde{f}(u)g(x)=\bigg(\frac{f(x)}{g(x)}\bigg)g(x)=f(x),
\end{displaymath}
and similarly for $v>\tilde{f}(u)$.
\end{remark}

Algorithm \ref{algo_inv_bisect_for_vN_non_compact} chooses randomly a uniform subinterval and, if permitted to run forever, would produce a random variable with distribution function $G$ restricted to \mbox{$[G^{-1}(u_1),G^{-1}(u_2)]$} as is illustrated in figure \ref{fig_bisect}. So, for random variate generation, we only replace line (\ref{bisecto_sub_routine_line_reject_algo}) of algorithm \ref{algo_reject_compact} by algorithm \ref{algo_inv_bisect_for_vN_non_compact} with $[u_1,u_2]=R^{*}$, and note that elsewhere in algorithm \ref{algo_reject_compact}, $f$ must be replaced by $\tilde{f}$.

\begin{theorem}\label{thm_gen_case_I}
Let $\epsilon>0$ and, for $j\in\mathbb{Z}$, let $I_j=[2\epsilon j, 2\epsilon(j+1))$. For $b>a$, recall that $F\big([a,b)\big)= F(b)-F(a)$. The expected number of bits used by algorithm \ref{algo_inv_bisect_for_vN_non_compact} is bounded above by
\begin{equation}
3+\sum_{j\in\mathbb{Z}}{F(I_j)\log_{2}\bigg(\frac{1}{F(I_j)}\bigg)},\label{sum_vn_non_compacte}
\end{equation}
If R\'{e}nyi's condition, as from section \ref{sect_cts_gen}, holds together with $H(f)>-\infty$, then, as $\epsilon \to 0^{+}$, the expected complexity is bounded above by
\begin{displaymath}
\log_{2}\bigg(\frac{1}{\epsilon}\bigg)+H(f)+5+o(1).
\end{displaymath}
\end{theorem}

Before proving theorem \ref{thm_gen_case_I}, here are a few remarks.
\begin{remark}
R\'{e}nyi's condition holds if and only if the sum (\ref{sum_vn_non_compacte}) is finite for $\epsilon=1$.
\end{remark}
\begin{remark}
Theorem \ref{thm_gen_case_I} establishes that algorithm \ref{algo_inv_bisect_for_vN_non_compact} is optimal to within an additive constant. In particular, its main term, $-\log_{2}(\epsilon)$, and second term, the differential entropy $H(f)$, match our lower bound from theorem \ref{lowerbndthm}.
\end{remark}

For convenience, our notation in the proof hereafter differs slightly from algorithm \ref{algo_inv_bisect_for_vN_non_compact} and illustration \ref{fig_bisect}. Indeed, we use two letters and one index, namely $[u_i,v_i]$ for some $i\in\mathbb{Z}$, to denote the $x$-coordinates of a rectangle while we used one letter and two indices, namely $[u_1,u_2]$, in algorithm \ref{algo_inv_bisect_for_vN_non_compact} and figure \ref{fig_bisect}.

\begin{proof}[Proof of theorem \ref{thm_gen_case_I}]
Let us denote an accepting rectangle by $R_i$ and its projection by $R_{i}^{\star}$. So, if $R_{i}^{\star}=[u_i,v_i]$, then $R_i=[u_i,v_i]\times [\alpha_i,\alpha_i+Cq_i]$, where $0\leq \alpha_i\leq \alpha_i+Cq_i\leq C$, $q_i\in[0,1]$. The probability mass of $R_i$ is given by $p_i= (v_i-u_i)C q_i$. By the mapping $G^{-1}$, $R_i$ gets mapped to a contiguous region $Q_i$ with projection $Q_i^{\star}= [a_i,b_i]$ such that
\begin{align*}
&a_i=G^{-1}(u_i),\quad b_i=G^{-1}(v_i),\quad\text{and thus }v_i-u_i = \int_{a_i}^{b_i}{g}=G(Q_i^{\star})=\frac{p_i}{C{}q_i}.
\end{align*}
We observe also that
\begin{displaymath}
\sum_{i:\, x\in Q_{i}^{\star}}{C q_i g(x)}=f(x)\quad\text{for all $x\in\mathbb{R}$.}
\end{displaymath}
We have that $\{I_j\}_{j\in\mathbb{Z}}$ defines a regular grid with intervals of length $2\epsilon$. Given some $i\in\mathbb{Z}$, algorithm \ref{algo_inv_bisect_for_vN_non_compact} on input $[u_i,v_i]$ takes at most
\begin{displaymath}
3+\sum_{j\in\mathbb{Z}}{\xi_{ji}\log_{2}\bigg(\frac{1}{\xi_{ji}}\bigg)}\quad\text{where}\quad
\xi_{ji}=\frac{G\big(I_j\cap Q_{i}^{\star}\big)}{G(Q_{i}^{\star})}\quad\text{and}\quad j\in\mathbb{Z}.
\end{displaymath}
Given some $i\in\mathbb{Z}$, it is important to observe that $(\xi_{ji})_{j\in\mathbb{Z}}=\mathbf{\xi}_{i}$ is a probability vector. Conditional on some $i\in\mathbb{Z}$, algorithm \ref{algo_inv_bisect_for_vN_non_compact} on inputs $[u_i,v_i]$ is equivalent to the algorithm by Han and Hoshi \cite{HanHos_1997} on input vector $\mathbf{\xi}_{i}$. To obtain unconditionnally an upper bound on the expected number of coin flips, we average over all $R_i$ which yields to
\begin{equation}
3+\sum_{i\in\mathbb{Z}}{p_i\sum_{j\in\mathbb{Z}}{\xi_{ji}\log_{2}\bigg(\frac{1}{\xi_{ji}}\bigg)}}\leq 3+\sum_{j\in\mathbb{Z}}\Bigg(\sum_{i\in\mathbb{Z}}{p_i\xi_{ji}}\Bigg)\log_{2}\Bigg(\frac{1}{\sum_{i\in\mathbb{Z}}{p_i\xi_{ji}}}\Bigg).\label{ikuiku_no1}
\end{equation}
The inequality from \ref{ikuiku_no1} is due to the concavity of $u\log_{2}\big(1\slash u\big)$ in $u$ and by Jensen's inequality.

We observe also that
\begin{align*}
&\sum_{i\in\mathbb{Z}}{p_i\xi_{ji}}=\sum_{i\in\mathbb{Z}}{p_i\frac{G\big(I_j\cap Q_{i}^{\star}\big)}{G(Q_{i}^{\star})}}=\sum_{i\in\mathbb{Z}}{C q_i G(I_j \cap Q_{i}^{\star})}\\
&\quad=\sum_{i\in\mathbb{Z}}{C q_i \int_{I_j}{\mathds{1}\{x\in Q_i^{\star}\}g(x){d}x}}=\int_{I_j}{\bigg(\sum_{i\in\mathbb{Z}}{C q_i \mathds{1}\{x\in Q_{i}^{\star}\}}\bigg)g(x){d}x}\\
&\quad=\int_{I_j}{f(x){d}x}= F(I_{j}),
\end{align*}
where $F$ is the distribution function of $f$. Thus the expected number of coin flips does not exceed
\begin{align}
&3+\sum_{j\in\mathbb{Z}}{F(I_j)\log_{2}\bigg(\frac{1}{F(I_j)}\bigg)}\label{ikuiku_no2}.
\end{align}
In (\ref{ikuiku_no2}), we recognize the entropy defined by the probability vector $\big(F(I_j)\big)_{j\in\mathbb{Z}}$. We recall again results from Csisz\'{a}r \cite{Csi_1961}, \cite{Csi_1962} which state that if $\big(F(I_j)\big)_{j\in\mathbb{Z}}$ has a finite entropy for some $\epsilon>0$, and if $H(f)>-\infty$, then
\begin{displaymath}
3+\sum_{j\in\mathbb{Z}}{F(I_j)\log_{2}\bigg(\frac{1}{F(I_j)}\bigg)}\leq H(f)+\log_{2}\frac{1}{\epsilon}+5+o(1)\quad\text{as $\epsilon \to 0^{+}$.}
\end{displaymath}
The discrete distribution $\big(F(I_j)\big)_{j\in\mathbb{Z}}$ has a finite entropy if and only if R\'{e}nyi's condition holds. The ``5'' can be replaced by ``3'' if in addition $f$ is bounded and decreasing on its support, $[0,\infty)$.
\end{proof}

\subsection{Convolutional sampling methods}\label{sect_convolution}

A convolutional sampling method is based on sampling the sum of independent random variables. Let $\{X_{i}\}_{i\in\mathbb{Z}}$ be a family of independently non-identically random variables, and let $I\subseteq \mathbb{Z}$, if $\sum_{i\in I}{X_{i}}$ converges in distribution to some random variable $X$, then we say that $X$ is the convolution of the $X_{i}$ for $i\in I$. The distribution of $X$ may be singular, absolutely continuous or discrete. For the case when $X$ is continuous, let $b>1$ be an integer, $\{X_{i}\}_{i\in\mathbb{Z}}$ as before, we are interested in cases when
\begin{displaymath}
X=\sum_{j=-\infty}^{n}{X_{j}b^{j}}\quad\text{for some $n\in\mathbb{Z}$.}
\end{displaymath}

When $b=2$, then $X_i$ is a Bernoulli random variable for $i\in I$. When $n=-1$ and $b=2$, Kakutani's result \cite{Kak_1948} characterizes the type of distribution that the $X_j$'s yield.
\begin{theorem}[Kakutani \cite{Kak_1948}]\label{thm_kakutani}
For all $i\in\mathbb{N}$, let $p_{i}\in[0,1]$ and let $X_{i}$ be independent Bernoulli random variables such that $\mathbf{P}\{X_{i}=1\}=p_{i}$. If $X=\sum_{i=1}^{\infty}{X_{i}2^{-i}}$, then
\begin{align*}
\text{$X$ is singular} &\Leftrightarrow \sum_{i=1}^{\infty}{\bigg(p_{i}-\frac{1}{2}\bigg)^{2}}\text{ diverges,}\\
\text{$X$ is absolutely continuous} &\Leftrightarrow \sum_{i=1}^{\infty}{\bigg(p_{i}-\frac{1}{2}\bigg)^{2}}\text{ converges,}\\
\text{$X$ is discrete} &\Leftrightarrow \prod_{i=1}^{\infty}{\bigg(\frac{1}{2}+\Big|p_{i}-\frac{1}{2}\Big|\bigg)}>0.
\end{align*}
\end{theorem}

First of all, as shown in \cite{Dev_book1986}, if $X$ is an exponential random variable with unit mean parameter, then $\lfloor X\rfloor$ is distributed as a geometric random variable with parameter $1\slash e$, and $X-\lfloor X\rfloor$, the fractional part of $X$, is distributed as a truncated exponential random variable on the interval $[0,1)$; moreover $\lfloor X\rfloor$ and $X-\lfloor X\rfloor$ are independent. We concentrate on the fractional part therefore. The following theorem tells us that the fractional part is the convolution of independent Bernoulli random variables.

\begin{theorem}\label{gravelconvolbernoexpo}
Let $(X_1,\ldots,X_j,\ldots)$ be a sequence of scaled independent Bernoulli distributed random variables with
\begin{displaymath}
\mathbf{P}\{X_j=2^{-j}\}=p_j=\frac{e^{-1\slash2^{j}}}{e^{-1\slash2^{j}}+1}\quad\text{and}\quad\mathbf{P}\{X_j=0\}=1-p_j\quad\text{for $j\in\mathbb{N}$.}
\end{displaymath}
If $X=\sum_{j=1}^{\infty}{X_j}$, then $X$ is a truncated exponential random variable, that is,
\begin{displaymath}
f(x)=\frac{e^{-x}}{1-e^{-1}}\textnormal{ for $x\in(0,1)$.}
\end{displaymath}
\end{theorem}
\begin{proof}[Proof of theorem \ref{gravelconvolbernoexpo}]
The Fourier transform of $X_j$ is
\begin{align*}
&\mathbf{E}(e^{\imath{}X_j{}t})=p_j{}e^{\imath{}t{}\slash2^{j}}+(1-p_j)=\frac{e^{((-1+\imath{}t)\slash2^{j})}+1}{e^{-1\slash2^{j}}+1}.
\end{align*}
Since $X$ is the sum of the independent $X_j$'s, we have
\begin{align*}
&\mathbf{E}(e^{\imath{}X{}t})=&\prod_{j=1}^{\infty}{\mathbf{E}(e^{\imath{}X_j{}t})}=\prod_{j=1}^{\infty}{\frac{e^{((-1+\imath{}t)\slash2^{j})}+1}{e^{-1\slash2^{j}}+1}}=\frac{1-e^{(-1+\imath{}t)}}{1-e^{-1}}\frac{\phantom{-}1}{-1+\imath{}t}\,
\end{align*}
which is the Fourier transform of $f(x)$.
\end{proof}
We can thus generate $X-\lfloor X\rfloor$ with precision $\epsilon$ if we set \mbox{$k=\lceil\log_{2}\big(\frac{1}{\epsilon}\big)\rceil$}, and let $Y=\sum_{j=1}^{k}{X_j}$. The differential entropy of $f$ is given by $H(f)=\frac{e}{e-1}\log_{2}(e-1)$. By using batch generation from section \ref{sect_batch_gen}, the expected complexity of the number random bits is asymptotically given by
$\log_{2}\big(\frac{1}{\epsilon}\big)+H(f)=\log_{2}\big(\frac{1}{\epsilon}\big)+1.2354\ldots$; we do not know currently a sampling algorithm with a better expected complexity.

\section{Conclusion and further research}\label{sect_conclusion}

We conclude our work with a few questions for further research and for which the source of randomness can be of another type than the type we have assumed which produces unbiased, i.i.d.\ random bits.
\begin{enumerate}
\item Besides the DDG-tree based algorithms based on mass functions (Knuth and Yao) or cumulative functions (Han and Hoshi), are there other natural DDG-tree based algorithms for discrete distributions? Perhaps following the line of research from section \ref{sect_generic_DDG}.
\item Is there a connection between an $(\epsilon, p)$-sampling graph as in definition \ref{defn_samplinggraph} from section \ref{sect_cts_gen} and transition graph of some probabilistic automata? What does the cycle structure of a sampling graph reveal, if it has cycles excluding loops?
\item Is it possible to further extend our rejection algorithms from section \ref{sect_vN} to densities that are not Riemann-integrable?
\item The study of upper bounds for singular distributions that cannot be expressed as a convolution as in section \ref{sect_convolution}.
\item Is there a general framework for mixtures of distributions which are neither continuous nor discrete?
\item Besides sampling distributions with coin flips, generating combinatorial objects from unbiased i.i.d.\ such as permutations as done in Bacher, Bodini, Hwang and Tsai \cite{BacBodHwaHsi_2017} would be very practical.
\end{enumerate}

\bibliographystyle{plain}
\newcommand{\SortNoop}[1]{}

\end{document}